\documentclass[dvipsnames]{article}

\usepackage{arxiv}
\usepackage[english]{babel}
\usepackage{authblk}
\usepackage{natbib}
\usepackage[utf8]{inputenc}
\usepackage[colorlinks=true,linkcolor=blue, citecolor=blue,backref=true]{hyperref}
\usepackage{acronym}
\usepackage{amsmath}
\usepackage{amssymb}
\usepackage{amsthm}
\usepackage{amsfonts}
\usepackage{mathtools}
\usepackage{xcolor}
\usepackage{tabularx}
\usepackage{booktabs}
\usepackage{tikz}
\usepackage{makecell}

\usepackage{subcaption}
\usepackage{algorithm}

\usepackage[noend]{algpseudocode}
\usepackage{graphicx}

\usepackage{outlines}
\usepackage{color,soul}

\newtheorem{theorem}{Theorem}

\newacro{mmd}[\textsc{mmd}]{maximum mean discrepancy}
\newacro{hsic}[\textsc{hsic}]{Hilbert-Schmidt independence criterion}
\newacro{hscic}[\textsc{hscic}]{Hilbert-Schmidt conditional independence criterion}
\newacro{rkhs}[\textsc{rkhs}]{reproducing kernel Hilbert space}
\newacro{anm}[\textsc{anm}]{additive noise model}
\newacro{sem}[\textsc{sem}]{structural equation model}
\newacro{dhsic}[{\small{$d-$}}\textsc{hsic}]{$d$-variable Hilbert-Schmidt independence criterion}
\newacro{kcit}[\textsc{kcit}]{kernel conditional independence test}
\newacro{kcipt}[\textsc{kcipt}]{kernel conditional independence permutation test}
\newacro{sdcit}[\textsc{sdcit}]{self-discrepancy conditional independence test}
\newacro{scm}[\textsc{scm}]{\textit{structural causal model}}
\newacro{kl}[\textsc{kl}]{Karhunen-Lo\`eve}
\newacro{se-t}[\textsc{se-t}]{squared-exponential $T$}
\newacro{se-c}[\textsc{se-c}]{squared-exponential $C$}
\newacro{gp}[\textsc{gp}]{Gaussian process}
\newacro{dag}[\textsc{dag}]{directed acyclic graph}
\newacro{cpdag}[\textsc{cpdag}]{completed partially directed acyclic graph}
\newacro{cme}[\textsc{cme}]{conditional mean embedding}
\newacro{cpt}[\textsc{cpt}]{conditional permutation test}
\newacro{ccm}[\textsc{ccm}]{convergent cross mapping}
\newacro{resit}[\textsc{resit}]{regression with subsequent independence test}
\newacro{shd}[\textsc{shd}]{structural Hamming distance}
\newacro{pcmci}[\textsc{pcmci}]{PC algorithm with momentary conditional independence}
\newacro{hflm}[\textsc{hflm}]{historical functional linear model}
\newacro{wgi}[\textsc{wgi}]{World Governance Indicator}

\newcommand\indep{\protect\mathpalette{\protect\independenT}{\perp}}
\def\independenT#1#2{\mathrel{\rlap{$#1#2$}\mkern2mu{#1#2}}}

\newcommand{\Y}{\mathbf{Y}}
\newcommand{\x}{\mathbf{x}}
\newcommand{\X}{\mathbf{X}}

\newcommand{\Z}{\mathbf{Z}}

\newcommand{\hsic}{\textsc{hsic}}
\newcommand{\dhsic}{\small{\emph{d-}}\textsc{hsic}}
\renewcommand{\P}{\mathbb{P}}
\renewcommand{\H}{\mathcal{H}}

\title{Kernel-based independence tests for causal structure learning on functional data}

\author[1]{Felix Laumann}
\author[2,3]{Julius von K\"ugelgen}
\author[2]{Junhyung Park}
\author[2]{Bernhard Sch\"olkopf}
\author[1]{Mauricio Barahona}

\affil[1]{Department of Mathematics, Imperial College London, London, UK}
\affil[2]{MPI for Intelligent Systems, T\"ubingen, Germany}
\affil[3]{Department of Engineering, University of Cambridge, Cambridge, UK}

\begin{document}
\everypar{\looseness=-1}

\maketitle

\begin{abstract}

Measurements of systems taken along a continuous functional dimension, such as time or space, are ubiquitous in many fields, from the physical and biological sciences to economics and engineering.
Such measurements can be viewed as realisations of an underlying smooth process sampled over the continuum. However, traditional methods for independence testing and causal learning are not directly applicable to such data, as they do not take into account the dependence along the functional dimension.
By using specifically designed kernels, we introduce statistical tests for bivariate, joint, and conditional independence for functional  variables.
Our method not only extends the applicability to functional data of the \ac{hsic} and its d-variate version ($d$-\textsc{hsic}), but also allows us to introduce a test for conditional independence by defining a novel statistic for the \ac{cpt} based on the \ac{hscic}, with    
optimised regularisation strength 
estimated through an evaluation rejection rate. 
Our empirical results of the size and power of these tests on synthetic functional data show good performance, and we then exemplify their application to several constraint- and regression-based causal structure learning problems, including both synthetic examples and real socio-economic data.

\end{abstract}

\section{Introduction} \label{sec:introduction}
Uncovering the causal relationships between measured variables, 
a discipline known as \textit{causal structure learning} or \textit{causal discovery}, is of great importance across various scientific fields, such as climatology~\citep{runge2019inferring}, economics~\citep{sulemana2018empirical}, or biology~\citep{finkle2018windowed}. Doing so from passively collected (`observational') data enables the inference of causal interactions between variables without performing experiments or randomised control trials, which 
are often expensive, unethical, or impossible to conduct \citep{glymour2019review}. %
Causal structure learning is the inference, under a given set of assumptions, of directed and undirected edges in graphs representing the data generating process, where the nodes represent variables and the inferred edges capture causal (directed) or non-causal (undirected) relationships between them.


Research in various areas collates \textit{functional} data consisting of multiple series of measurements observed conjointly over a given continuum (e.g., time, space, or frequency), where each series is assumed to be a realisation of an underlying \textit{smooth} process \citep[\S 3]{ramsay2005}.
By viewing the series of measurements as discretisations of functions, the observations are not required to be collected over regular meshes of points along the continuum.
If the variables are measured over time as the underlying continuum, there is a long history of methods that have been developed to infer (time-based) causality between variables. 
Among those, the classic Granger-causality~\citep{granger1969investigating} declares that variable `$X$ causes $Y$' ($X \to  Y$) if predicting the future of $Y$ becomes more accurate with, as compared to without, access to the past of $X$, conditional on all other relevant variables \citep{geweke1982measurement}.
However, these methods 
assume that the observed time-series are stationary and the causal dependency of $X$ on $Y$ is linear. 
More recently, \citet{sugihara2012detecting}~developed \acp{ccm}, a method that relaxes the assumption of linearity and finds causal relationships based on time-embeddings of the (stationary) time-series at each point. 
While useful in many situations, Granger-causality and \ac{ccm} can perform weakly when the time-series for $X$ and $Y$ are nonlinearly related or nonstationary, respectively
(see Appendix~\ref{app:linear_stationary}).

Here we present a method that uses kernel-based independence tests to detect statistically significant causal relationships by extending constraint- and regression-based causal structure learning to functional data. 
The key advantages over Granger-causality and \ac{ccm} are both the systematic consideration of confounders and the relaxation of assumptions around linear relationships or stationarity in the data, which can lead to different causal relationships between variables.
As a motivating example, consider the relationship between two variables: `corruption' and `income inequality', as measured by the \acp{wgi}~\citep{kaufmann2011worldwide} and the \citet{worldbank2022}, respectively. 
Using data for 
48 African countries from 1996 to 2016, \citet{sulemana2018empirical} investigated their cause-effect relationship 
and found that corruption `Granger-causes' income inequality. 
We have also confirmed independently that applying \ac{ccm} to the same data leads to the same conclusion.
However, by considering the time-series data as realisations of functions over time, and thus avoiding linearity and stationarity assumptions, our proposed kernel-based approach suggests the reverse result, i.e., causal influence of income inequality on corruption appears as the more statistically likely direction. 
Although a bidirectional causal dependency between these two variables might appear as more realistic, this conclusion is in agreement with other quantitative findings, which draw on different data sources \citep{jong2005comparative,alesina2005corruption,dobson2010there}. 
We will return to this example in Section~\ref{sec:wgi_data} where we analyse causal dependencies between all six \acp{wgi}.


Methodologically, our work extends the applicability of two popular paradigms in causal structure learning---\textit{constraint-based}~\citep[\S~5]{spirtes2000causation} and \textit{regression-based} methods~\citep{peters2014causal}---to functional data. 
Independence tests play a crucial role in uncovering causal relationships in both paradigms, and kernels provide a powerful framework for such tests by embedding probability distributions in \acp{rkhs} \citep[\S~2.2]{muandet2017kernel}.
Until now, however, related methods for causal learning had only been applicable 
to univariate and multivariate data, but not to functional data.
To address this limitation, we employ recently derived kernels over functions~\citep{wynne2022kernel} 
to widen the applicability of kernel-based independence tests to functional data settings. 
To test for conditional independence, we can then compute \ac{hscic}~\citep{park2020measure} 
in a \acl{cpt} (\ac{cpt})~\citep{berrett2020conditional}, and we 
propose a straightforward search to determine the optimised regularisation rate in \ac{hscic}. 

We structure our paper as follows. 
Section~\ref{sec:background} provides a brief overview of prior literature on functional data analysis, kernel-based independence tests, and causal structure learning methods. 
Section~\ref{sec:methods} presents the definition of a conditional independence test for functional data and its applicability to causal structure learning on such data. 
We then empirically analyse the performance of our independence tests and causal structure learning algorithms on synthetic and real-world data in Section~\ref{sec:experiments}. 
We conclude with a discussion in Section~\ref{sec:discussion}. 

Our main contribution lies in Section~\ref{sec:methods} where we propose a conditional independence test for functional data that combines a novel test statistic based on \ac{hscic} with \ac{cpt} to generate samples under the null hypothesis.
The algorithm also searches for the optimised regularisation strength $\lambda^*$ required to compute \ac{hscic}, by pre-test permutations to calculate an \emph{evaluation rejection rate}.
We also highlight the following secondary contributions:
\begin{itemize}
    \item In Section~\ref{sec:joint_experiments}, we extend the historical functional linear model \citep{malfait2003historical} 
    to the multivariate case $\{X_1, \dots, X_i, \dots, X_d\} \rightarrow Y$ for regression-based causal structure learning, and we show how a joint independence test can be used to verify candidate \acp{dag}~\citep[\S~5.2]{pfister2018kernel} that embed the causal structure of function-valued random variables. 
    This model has been contributed to the Python package \texttt{scikit-fda}~\citep{ramos2019scikit}.
    \item On synthetic data, we show empirically that our bivariate, joint and conditional independence tests achieve high test power, and that our causal structure learning algorithms outperform previously proposed methods. 
    \item Using a real-world data set (World Governance Indicators) we demonstrate how our method can yield insights into cause-effect relationships amongst socio-economic variables measured in countries worldwide.
    \item Implementations of our algorithms are made available at \url{https://github.com/felix-laumann/causal-fda/} in an easily usable format that builds on top of \texttt{scikit-fda} and \texttt{causaldag}~\citep{squires2018causaldag}.
\end{itemize}

\section{Background and Related Work}
\label{sec:background}

\subsection{Functional Data Analysis}
In functional data analysis~\citep{ramsay2005}, a variable $X$ is described by a set of $n$ samples (or realisations), $\X = \{x_i(s)\}_{i=1}^n$, where each functional sample $x_i(s)$ corresponds to a series of observations over the continuum $s$, also called the \textit{functional dimension}. Typical functional dimensions are time or space.
In practical settings, the observations are taken at a set of $S$ discrete values $s_1, \ldots, s_S$ of the continuum variable $s$.
Examples of functional data sets include the vertical position of the lower lip over time when speaking out a given word~\citep{malfait2003historical}, the muscle soreness over the duration of a tennis match~\citep{girard2006changes}, or the ultrafine particle concentration in air measured over the distance to the nearest motorway~\citep{zhu2006comparison}.

In applications, the functional samples are usually represented as linear combinations of a finite set of $M$ basis functions $\{\phi_m(s)\}_{m=1}^M$ (e.g., Fourier or monomial basis functions):
\begin{equation}
\label{eq:basis_representation}
    x_i(s) = \sum_{m=1}^M c_{i,m} \phi_m(s),
\end{equation}
where the coefficients $C_i = (c_{i,1}, c_{i,2}, \dots, c_{i,M})$ characterise each sample.
If the number of basis functions is equal to the number of observations, $M = S$, each observed value $x_i(s_k)$ can be fitted exactly by obtaining the coefficients $C_i$ using standard nonlinear least squares fitting techniques (provided the  $\phi_m(s)$ are valid basis functions), and Equation~\eqref{eq:basis_representation} allows us to interpolate between any two observations. 
When the number of basis functions is smaller than the number of observations, $M < S$, as it is commonly the case in practice, the basis function expansion~\eqref{eq:basis_representation} provides a smoothed approximation to the set of observations, $\hat{x}_i(s_k)$.

%
%

For the many applications where the continuum is time, \acp{hflm}~\citep{malfait2003historical} provide a comprehensive framework to map the relationship between two sets of functional samples. 
Let $0$ and $T$ be the initial and final time points for a set of samples $y_i(t)$. 
\acp{hflm} describe dependencies that can vary over time using the function $\beta(s, t)$, which encapsulates the influence of $x(s), s \in [0, S]$ on another variable $y(t), t \in [0, T]$ at any two points in time, $s_k$ and $t_k$:
\begin{equation}\label{eq:hist_data}
    y_i (t) = \int_{s_0(t)}^t x_i(s) \; \beta(s,t) \; ds \; , \quad t \in [0, T] 
\end{equation}
where $s_0(t)$ is the maximum allowable lag for any influence of $X$ on $Y$ and $s_0(t) \leq s \leq t$. 
Typical choices for $\beta(s,t)$ are exponential decay and hyperbolic paraboloid (or ``saddle'') functions.
The continuum is not required to be time but can also be space, frequency or others (see \citet{ramsay2017dynamic} for an extensive collection of function-to-function models and applications).

\subsection{Kernel Independence Tests}
Let $\mathcal{H}_{k_1}$ and $\mathcal{H}_{k_2}$ be separable \acp{rkhs} with kernels $k_1:\mathbb{R}^\mathbf{X}\times\mathbb{R}^\mathbf{X}\rightarrow\mathbb{R}$ and $k_2:\mathbb{R}^\mathbf{Y}\times\mathbb{R}^\mathbf{Y}\rightarrow\mathbb{R}$ such that the tensor product kernel $(k_1\otimes k_2) (x,y,x',y')= k_1(x,x')\cdot k_2(y,y')$ implies $(\mathbb{R}^\mathbf{X}\times\mathbb{R}^\mathbf{Y})\times(\mathbb{R}^\mathbf{X}\times\mathbb{R}^\mathbf{Y})\rightarrow\mathbb{R}$. 
If the kernel $k_1$ uniquely embeds a distribution $\P_{\X}$ in an \ac{rkhs} by a mean embedding,
\begin{equation}
    \mu_{\X} := \int_{\mathcal{X}} k_1(\x, \cdot) d\P_{\X}
\end{equation}
which captures any information about $\P_\X$, we call $k_1$ a characteristic kernel on $\mathbb{R}^\mathbf{X}\times\mathbb{R}^\mathbf{Y}$ \citep{fukumizu2007kernel}.
Characteristic kernels have thus been extensively used in bivariate ($\P_{XY} = \P_X \P_Y$), joint ($\P_{XYZ} = \P_X \P_Y \P_Z$) and conditional ($\P_{XY|Z} = \P_{X|Z} \P_{Y|Z}$) independence tests~\citep[e.g.,][]{gretton2008kernel, pfister2018kernel, zhang2012kernel}.

For the bivariate independence test, let $\P_{\X\Y}$ denote the joint distribution of $\X$ and $\Y$. Then 
\acs{hsic} is defined as
\begin{align}
    \hsic(\H_{k_1},  \H_{k_2}, \P_{\X \Y}) &:= \| \mu_{\X\Y} - \mu_{\X} \otimes \mu_{\Y} \|_{\H_{k_1} \otimes \H_{k_2}}^2 \geq 0, \\ 
    & \text{ with equality iff} \quad \P_{\X\Y} = \P_\X \P_\Y \;  . \nonumber
\end{align}
We refer to \citet{gretton2008kernel} for the definition of an estimator for finite samples that constitutes the test statistic in the bivariate independence test with null hypothesis $H_0 : X \indep Y$. The test statistic is then computed on the original data $(\X, \Y)$ and statistically compared to random permutations $\{(\X, \Y_{p'})\}_{p'=1}^P$ under the null hypothesis.

For distributions with more than two variables, let $\P_{\X_1, \X_2, \dots, \X_d}$ denote the joint distribution on $\X_1, \X_2, \ldots, \X_d$. To test for joint independence we compute 
\begin{align}
    \dhsic(\H_{k_1},  \H_{k_2}, \dots, &\H_{k_d}, \P_{\X_1, \X_2, \dots, \X_d}) := \\ 
    &\| \mu_{\X_1,\X_2,\dots,\X_d} - \mu_{\X_1} \otimes \mu_{\X_2} \otimes \dots \otimes \mu_{\X_d} \|_{\H_{k_1} \otimes \H_{k_2} \otimes \dots \otimes \H_{k_d}}^2 \geq 0 \nonumber \\
    &\text{with equality iff} \quad \P_{\X_1,\X_2, \dots, \X_d} = \P_{\X_1} \P_{\X_2} \cdots \P_{\X_d} \; .
    \nonumber
\end{align}
\citet{pfister2018kernel} derive a numerical estimator, which serves as the basis for a joint independence test on finite samples. Here, the distribution under the null hypothesis of joint independence is generated by randomly permuting all sample sets in the same way as $\Y$ is in the bivariate independence test.

Lastly, the conditional independence test relies on accurately sampling from the distribution under the null hypothesis $H_0 : X \indep Y | Z$. At the core of the 
\acl{cpt}~\citep{berrett2020conditional} lies a randomisation procedure that generates permutations of $\X$, denoted $\{\X_{p'}\}_{p'=1}^P$, which are generated without altering the conditional distribution $\P_{\X | \Z}$, so that
\begin{equation}
    \P_{\X_{p'}, \Y, \Z} = \P_{\X, \Y, \Z}
\end{equation}
under $H_0$ while breaking any dependence between $X$ and $Y$.
The null distribution can therefore be generated by repeating this procedure multiple times, and we can decide whether $H_0$ should be rejected by comparing a test statistic on the original data against its results on the generated null distribution.

The existing literature on kernel-based independence tests is extensive, see e.g., \citet{berrett2020conditional} for a relevant review, but only a small part of those tests investigates independence among functional variables~\citep{lai2021testing,gorecki2020independence}. 
There have been particularly strong efforts in developing conditional independence tests and understanding their applicable settings. The \ac{kcit}~\citep{zhang2012kernel}, for example, gave promising results in use with univariate data but increasingly suffered when the number of conditional variables was large. 
In contrast, the \ac{kcipt}~\citep{doran2014permutation} repurposed the well-established kernel two-sample test~\citep{gretton2012kernel} to a conditional independence setting which delivered stable results for multiple conditional variables. 
However, \citet{lee2017self} pointed out that as the number of permutations increases, while its power increases, its calibratedness decreases. 
This issue was overcome by their proposed \ac{sdcit}, which is based on a modified unbiased estimate of the \ac{mmd}.

\subsection{Causal Structure Learning}
The aim of causal structure learning, or causal discovery, is to infer the qualitative causal relationships among a set of observed variables, typically in the form of a causal diagram or \ac{dag}.
Once learnt, such a causal structure can then be used to construct a causal model such as a causal Bayesian network or a \ac{scm}~\citep{pearl2009causality}.
Causal models are endowed with a notion of manipulation and, unlike a statistical model, do not just describe a single distribution, but many distributions indexed by different interventions and counterfactuals.
They can be used for causal reasoning, that is, to answer causal questions such as computing the average causal effect of a treatment  on a given outcome variable. Such questions are of interest across many disciplines, and causal discovery is thus a highly topical area. We refer to~\citet{mooij2016distinguishing,peters2017elements,glymour2019review,scholkopf2022statistical,squires2022causal,vowels2022d} for comprehensive surveys and accounts of the main research concepts.
In particular, we focus here on causal discovery methods for causally sufficient systems, for which there are no unobserved confounders influencing two or more of the observed variables. 
Existing causal discovery methods can roughly be categorised into three families: 
\begin{itemize}
    \item  Score-based approaches assign a score, such as a penalised likelihood, to each candidate graph and then pick the highest scoring graph(s).
    A common drawback of score-based approaches is the need for a combinatorial enumeration of all \acp{dag} in the optimisation, although greedy approaches have been proposed to alleviate such issues~\citep{chickering2002optimal}.
    \item  Constraint-based methods start by characterising the set of \textit{conditional independences} in the observed data~\citep{spirtes2000causation}. They then determine the graph(s) consistent with the detected conditional independences by using a graphical criterion called d-separation, as well as the causal Markov and faithfulness assumptions, which establish a one-to-one connection between d-separation and conditional independence (see Appendix~\ref{app:defs} for definitions).
    When only observational i.i.d.\ data are available, this yields a so-called Markov equivalence class, possibly containing multiple candidate graphs. For example, the graphs $X\to Y\to Z$, $X \leftarrow Y \leftarrow Z$, and $X \leftarrow Y \to Z$ are Markov equivalent, as they all imply $X\indep Z | Y$ and no other conditional independence relations.
    \item Regression-based approaches directly fit the structural equations $X_i := f_i(\mathbf{PA}_i,U_i)$ of an underlying \ac{scm} for each $X_i$, where $\mathbf{PA}_i$ denote the parents of $X_i$ in the causal \ac{dag} and $U=(U_1, ..., U_n)$ are \textit{jointly independent} exogenous noise variables.
    Provided that the function class of the $f_i$ is sufficiently restricted, e.g., by considering only linear relationships~\citep{shimizu2006linear} or additive noise models~\citep{hoyer2008nonlinear}, the true causal graph is identified as the unique choice of parents for each $i$ such that the resulting residuals $\hat U_i=X_i-\hat f_i(\mathbf{PA}_i)$ are jointly independent.
\end{itemize}
As can be seen from these definitions, conditional, bivariate, and joint independence tests are an integral part of constraint- and regression-based causal discovery methods. 
Our main focus in the present work is therefore to extend the applicability of these causal discovery frameworks to functional data by generalising the underlying independence tests to such domains.



\section{Methods} \label{sec:methods}
In all three of our independence tests (bivariate, joint, conditional), we employ kernels over functions, also known as \ac{se-t} kernels~\citep{wynne2022kernel}.
Let $\mathcal{X}$ and $\mathcal{Y}$ be real, separable Hilbert spaces with norms $\| \cdot \|_{\mathcal{X}}$ and $\| \cdot \|_{\mathcal{Y}}$, respectively.
Then, for $T: \mathcal{X} \rightarrow \mathcal{Y}$, \ac{se-t} kernels are defined as
\begin{equation}
    k_T(x,y) = e^{-\frac{1}{2 \sigma^2} \| T(x) - T(y) \|_{\mathcal{Y}}^2} \; .
\end{equation}
where $\sigma^2$ is commonly defined as the median heuristic, $\sigma^2 = \text{Median}\{ \| T(a) - T(b) \|_{\mathcal{Y}}^2 : a,b \in \{ x_i \}_{i=1}^n \cup \{ y_i \}_{i=1}^n, a \neq b \}$. 
Replacing any characteristic kernel by the \ac{se-t} kernels for bivariate independence tests (based on \ac{hsic}) and for joint independence tests (based on \ac{dhsic}) is straightforward and does not require further theoretical investigation besides evaluating numerically the validity and power of the tests (see Section~\ref{sec:experiments}). However, the application of \ac{se-t} kernels to conditional independence tests needs further theoretical results, as we discuss now.

\subsection{Conditional independence test on functional data}

We consider the conditional independence test, which generates samples under the null hypothesis based on the \ac{cpt}, and uses the sum of \acp{hscic} over all samples $z \in \Z$ as its test statistic. 
The \ac{cpt} defines a permutation procedure that preserves the dependence of $Z$ on both $X$ and $Y$ while resampling data for $X$ that eliminates any potential dependence between $X$ and $Y$. 
This procedure results in samples according to the null hypothesis, $X \indep Y | Z$. 
We use this procedure whenever permutations are required as part of the conditional independence test. 
Given the computation of the \ac{hscic} is based on a kernel ridge regression, it requires to set a regularisation strength $\lambda$, which has been manually chosen previously~\citep{park2020measure}. 
Generally, our aim is to have type-I error rates close to the allowable false positive rate $\alpha$. 
However, choosing $\lambda$ inappropriately may result in an invalid test (type-I error rates exceed $\alpha$ if $\lambda$ is chosen too large), or in a deflated test power (type-I error rates are well below $\alpha$ and type-II error rates are high if $\lambda$ is chosen too small). 
Thus, we must define an algorithm that conducts a search over a range of potentially suitable values for $\lambda$ and assesses each candidate value by, what we will call, an \textit{evaluation rejection rate}.

The search proceeds by iterating over a range of values $\{\lambda_l\}_{l=1}^L$ to find the optimised value $\lambda^*$, as follows. 
For each $\lambda_l$, we start by producing one permutation of the samples $\X$, which we denote as $\X_b$, and we compute its corresponding \emph{evaluation test statistic} given by the sum of \acp{hscic} over all samples $z \in \Z$.
We then apply the usual strategy for permutation-based statistical tests: we produce an additional set of $P$ permuted sample sets of $\X$, which we denote $\{\X_{\pi}\}_{\pi=1}^P$, and for each $\X_{\pi}$, we determine the sum of \acp{hscic} over $z \in \Z$ to generate a distribution over statistics under the null hypothesis, which we call the \textit{evaluation null statistics}. 
We then compute the percentile where the evaluation test statistic on $\X_b$ falls within the distribution of evaluation null statistics on the permutation set $\{\X_{\pi}\}_{\pi=1}^P$. 
This results in an evaluation p-value which is compared to the allowable false positive rate $\alpha$ to establish
if the hypothesis $H_0: \X_b \not\indep \Y|\Z$ can be rejected. 
Given that both the evaluation test statistic and the evaluation null statistics are computed on conditionally independent samples, we repeat this procedure for $b=1, \ldots, B\geq100$ times to estimate an \emph{evaluation rejection rate} for each value of  $\lambda_l$.
Having completed this procedure over all values $\{\lambda_l\}_{l=1}^L$, we select the $\lambda_l$ that produces an evaluation rejection rate closest to $\alpha$ as the optimised regularisation strength, $\lambda^*$. 
Finally, we apply a \ac{cpt}-based conditional independence test using the optimised $\lambda^*$ to test the null hypothesis $H_0: X \not\indep Y | Z$. 
This entire procedure is summarised in Algorithm~\ref{alg:search_algorithm}. 

\begin{algorithm}
\caption{Search for $\lambda^*$ with subsequent conditional independence test}\label{alg:search_algorithm}
\begin{algorithmic}[1]
\Require{Samples ($\X, \Y, \Z$), Range $\{\lambda_l\}_{l=1}^L$, Significance level $\alpha$, Permutation iterations $P$, Rejection iterations $B$}
\Ensure{$\text{evaluation rejection rate}[\lambda^*] = 0$}
\For{$1 \leq l \leq L$}  \Comment{\textbf{Start:} Search for $\lambda^*$}
\For{$1 \leq b \leq B$}
\State Permute $\X$ by \ac{cpt}, call them $\X_b$
\State $\text{evaluation test statistic}[b] \gets \sum_{z \in \Z}\ac{hscic} (\X_b, \Y, \Z, \lambda_l)$
\For{$1 \leq \pi \leq P$}   
\State Permute $\X_b$ by \ac{cpt}, call them $\X_{b,\pi}$
\State $\text{evaluation null statistics}[b,\pi] \gets \sum_{z \in \Z}\ac{hscic} (\X_{b,\pi}, \Y, \Z, \lambda_l)$
\EndFor 
\State $\text{evaluation p-value} \gets 1-$ 
\State \qquad \qquad $\text{percentile}(\text{evaluation test statistic}[b], \text{evaluation null statistics}[b,:])$
\If{$\text{evaluation p-value} \geq \alpha$}
\State Fail to reject $H_0: \X_b \indep \Y|\Z$
\State $\text{rejects}[b]=0$
\Else
\State Reject $H_0: \X_b \not\indep \Y|\Z$
\State $\text{rejects}[b]=1$
\EndIf
\EndFor
\State $\text{evaluation rejection rate}[\lambda_l] \gets \text{mean}_{1 \leq b \leq B} (\text{rejects})$
\If{$|\text{evaluation rejection rate}[\lambda^*] - \alpha| \geq |\text{evaluation rejection rate}[\lambda_l] - \alpha|$}
\State $\lambda^* \gets \lambda_l$  
\EndIf
\EndFor
\State \textbf{return} $\lambda^*$  
\Comment{\textbf{End:} Search for $\lambda^*$}
\vspace{.3in}
\State $\text{test statistic} \gets \sum_{z \in \Z}\ac{hscic}(\X, \Y, \Z, \lambda^*)$   \Comment{\textbf{Start:} Conditional independence test}
\For{$1 \leq p' \leq P$}   
\State Permute $\X$ by \ac{cpt}, call them $\X_{p'}$
\State $\text{null statistics}[p'] \gets \sum_{z \in \Z}\ac{hscic} (\X_{p'}, \Y, \Z, \lambda^*)$
\EndFor 
\State $\text{p-value} \gets 1- \text{percentile}(\text{test statistic}, \text{null statistics})$ 
\If{$\text{p-value} \geq \alpha$}
\State Fail to reject $H_0: X \indep Y|Z$
\Else
\State Reject $H_0: X \not\indep Y|Z$
\EndIf  \Comment{\textbf{End:} Conditional independence test}
\end{algorithmic}
\end{algorithm}

The following Theorem~\ref{th:hscic_consistent} guarantees the consistency of the conditional independence test in Algorithm~\ref{alg:search_algorithm} with respect to the regularisation parameter $\lambda$.
\begin{theorem} \label{th:hscic_consistent}
		Let \(\mathcal{H}_{k_1}\) and \(\mathcal{H}_{k_2}\) be separable RKHSs with kernels \(k_1:\mathbb{R}^\mathbf{X}\times\mathbb{R}^\mathbf{X}\rightarrow\mathbb{R}\) and \(k_2:\mathbb{R}^\mathbf{Y}\times\mathbb{R}^\mathbf{Y}\rightarrow\mathbb{R}\) such that the tensor product kernel \(k_1\otimes k_2:(\mathbb{R}^\mathbf{X}\times\mathbb{R}^\mathbf{Y})\times(\mathbb{R}^\mathbf{X}\times\mathbb{R}^\mathbf{Y})\rightarrow\mathbb{R}\) is a characteristic kernel on \(\mathbb{R}^\mathbf{X}\times\mathbb{R}^\mathbf{Y}\)~\citep{szabo2017characteristic}. If the regularisation parameter \(\lambda=\lambda_n\) decays as \(n\rightarrow\infty\) at a slower rate than \(n^{-1/2}\), then the test based on the test statistic in Algorithm~\ref{alg:search_algorithm}  (lines 19--27) is consistent. 
\end{theorem}
\begin{proof}
See Appendix~\ref{app:proof}. 
\end{proof}

\subsection{Causal structure learning on functional data}

To infer the existence of (directed) edges in an underlying causal graph $G$ based on the joint data distribution $\P$ over a set of observed variables, we must assume that $G$ and $\P$ are intrinsically linked.
\citet[\S~2.3.3]{spirtes2000causation} define the \emph{faithfulness} assumption, from which it follows that if two random variables are (conditionally) independent in the observed distribution $\P$, then they are d-separated in the underlying causal graph $G$ \citep[\S~1.2.3]{pearl2009causality}.
Using this fact, constraint-based causal structure learning methods \citep[\S~5]{spirtes2000causation} take advantage of bivariate and conditional independence tests to infer whether two nodes are d-separated in the underlying graph.
These methods yield \acp{cpdag}, which are graphs with undirected and/or directed edges. 
In contrast, regression-based methods \citep{peters2014causal}, utilise the joint independence test to discover \acp{dag} that have all edges oriented. 
We now describe both of these approaches in more detail.

\subsubsection{Constraint-based causal structure learning}

Constraint-based causal structure learning relies on performing conditional independence tests for each pair of variables, $X$ and $Y$, conditioned on any possible subset of the remaining variables, that is, any subset within the power set of the $d \setminus \{X, Y\}$ remaining variables. 
Therefore, for $d$ variables we need to carry out $2^{d-2}$ conditional independence tests for every pair of variables, which results in $\binom{d}{2} 2^{d-2}$ tests in total.

Conditional dependences found in the data can then be used to delete and orient edges in the graph $G$, as follows.
We start with a complete graph with $d$ nodes, where each node corresponds to one of the variables. 
The edge connecting $X$ and $Y$ is deleted if there exists a subset $\mathcal{Z}_1$ of the $d-2$ remaining variables such that 
$X \indep Y | \mathcal{Z}_1$.
If we then find another subset $\mathcal{Z}_2$, such that $Z_* \in \mathcal{Z}_2$, $\mathcal{Z}_1 = \mathcal{Z}_2 \setminus \{Z_*\}$, and for which $X \not\indep Y | \mathcal{Z}_2$, we can orient the edges $X-Z_*$ and $Y-Z_*$ to form colliders (or v-structures), $X \rightarrow Z_* \leftarrow Y$.

Based on these oriented edges, `Meek's orientation rules' defined by \citet{meek1995complete} \citep[see e.g.,][\S~7.2.1]{peters2017elements} can be applied to direct additional edges based on certain graphical compositions, as shown in Appendix~\ref{app:alg+meek_rules}.
Briefly, these rules follow from the avoidance of cycles, which would violate the acyclicity of a \ac{dag}, and colliders, which violate the found conditional independence. 
Algorithms that implement these criteria are the SGS algorithm and the more efficient PC algorithm~\citep{spirtes2000causation}, which we summarise in Appendix~\ref{app:alg+meek_rules}.

\subsubsection{Regression-based causal structure learning}
To carry out our regression-based causal learning, we choose
\acp{anm}, a special class of \acsp{scm} where the noise terms are assumed to be independent and additive. 
This assumption guarantees that the causal graph can be identified if the function $f_i$ is nonlinear with Gaussian additive noise \citep{hoyer2008nonlinear,peters2012identifiability}, a typical setting in functional data. 
Our model over the set of random variables 
$\{X\}_{i=1}^d$
is given by:
\begin{equation}
    \X_i := f_i(\mathbf{PA}_i) + U_i 
    \; ,
\end{equation}
where the additive noise terms $U_i$ are jointly independent, i.e., any noise term $U_i$ is independent of any intersection of the other noise terms, 
\begin{equation*}
    \P(\cap_{j=1}^k U_{i_j}) = \prod_{j=1}^k \P(U_{i_j}), \quad
    \forall \ k \leq d .
\end{equation*}
Based on these assumptions, we follow \textsc{resit} \citep[Regression with Subsequent Independence Testing;][]{peters2014causal}, an approach to causal discovery that can be briefly described as follows. 
If we only have two variables $X$ and $Y$, we start by assuming that the samples in $\Y$ are the effect and the samples in $\X$ are the cause. Therefore we write $\Y$ as a function $\hat{f}_\Y$ of $\X$ plus some noise, and we test whether the residual $r_\Y = \Y - \hat{f}_\Y(\X)$ is independent of $\X$. 
We then exchange the roles of assumed effect and assumed cause to obtain the residual $r_\X = \X - \hat{f}_\X(\Y)$, which is tested for independence from $\Y$.
To overcome issues with finite samples and test power, \citet{peters2014causal} used the p-value of the independence tests as a measure of strength of independence, which we follow in our experiments in Section~\ref{sec:experiments}. 
Alternatively, one may determine the causal direction for this bivariate case by comparing both cause-effect directions with respect to a score defined by \citet{buhlmann2014cam}.

For $d$ variables, the joint independence test replaces the bivariate independence test. 
Firstly, we consider the set of candidate causal graphs, $\Gamma = \{G_1, \dots, G_\gamma, \dots, G_\Gamma\}$, which contains every potential 
\ac{dag} with $d$ nodes. 
For each candidate \ac{dag}, $G_\gamma$, we regress each descendant $\X_i$ on its parents $\mathbf{PA}_i$:
\begin{equation}\label{eq:add_noise_model}
    \X_i := \sum_{k \in \mathbf{PA}_i} \hat{f}_{i, k}(\X_k) + U_i \; , \qquad i \in {1, \dots, d} \; ,
\end{equation}
and we compute the residuals $r_{\X_i} = \X_i - \sum_{k \in \mathbf{PA}_i} \hat{f}_{i, k}(\X_k)$. We then apply the joint independence test to all $d$ residuals. 
The candidate \ac{dag} $G_\gamma$ is accepted as the `true' causal graph if the null hypothesis of joint independence amongst the residuals is not rejected. 
This process is repeated over all candidate causal graphs in the set $\Gamma$.
Again, because in finite samples this procedure may not always lead to only one candidate \ac{dag} being accepted, the one with the highest p-value is chosen~\citep{peters2014causal}.

\section{Experiments} \label{sec:experiments}

In practice, functional data analysis requires samples of continuous functions evaluated over a mesh of discrete observations. 
In our synthetic data, we consider the interval $s = [0, 1]$ for our functional dimension and draw 100 equally spaced points over $s$ to evaluate the functions. 
For real-world data, we map the space in which the data live (e.g., the years 1996 to 2020) to the interval $s = [0, 1]$ and interpolate over the discrete measurement points. We then use this interpolation to evaluate the functional samples on 100 equally spaced points.
Unless otherwise mentioned, henceforth we use \ac{se-t} kernels with $T = I$ where $I$ is the identity matrix and $\sigma^2$ as the median heuristic, as previously defined in Section~\ref{sec:methods}.

\subsection{Evaluation of independence tests for functional data} 
\label{sec:independence_tests}

Before applying our proposed independence tests to causal structure learning problems, we use appropriate synthetic data to evaluate their type-I error rate (test size) and type-II error rate (test power). 
Type-I errors are made when the samples are independent (achieved by setting $a=0$ in our experiments below) but the test concludes that they are dependent, i.e., the test falsely rejects the null hypothesis $H_0$. 
In contrast, type-II errors appear when the samples are dependent but the test fails to reject $H_0$ (achieved by setting $a>0$).
Specifically, we compute the error rates on 200 independent trials, which correspond to different random realisations of the data sets, i.e., synthetically generated data sets with random Fourier basis coefficients, random coefficients of the $\beta$-functions and random additive noise. 
We set the test significance level at $\alpha = 0.05$ and approximate the null hypothesis of independence by 1000 permutations using the \ac{cpt} scheme described above. 
%
Note that although our synthetic data are designed to replicate typical behaviour in time-series, the independence tests are not limited to time as the functional dimension, and can be used more generally.

\subsubsection{Bivariate independence test}\label{sec:marg_experiments}

We consider $n$ functional samples $\{x_i(s)\}_{i=1}^n$ of a random variable $X$ defined over the interval $s = [0, 1]$. 
To generate our samples, we sum $M = 3$ Fourier basis functions $\phi_{m,T} (s)$ with period $T = 0.1$ and coefficients $c_{i,m}$ randomly drawn from a standard normal distribution:
\begin{equation} \label{eq:func_data}
    x_i (s) = \sum_{m=1}^M c_{i,m} \phi_{m,T} (s) \; ,
\end{equation}
where the Fourier functions are $\phi_{1,T}(s) = 1$, $\phi_{2, T} (s) = \sqrt{\frac{2}{T}} \sin \left(\frac{2\pi m}{T}s \right)$ and $
    \phi_{3, T} (s) = \sqrt{\frac{2}{T}} \cos \left(\frac{2\pi m}{T}s \right)$. 
To mimic real-world measurements, we draw each sample over a distinct, irregular mesh of evaluation points. We then interpolate them with splines, and generate realisations over a regular mesh of points. 
Multivariate random noise $\epsilon_i^X (s) \sim \mathcal{N}(0, I)$, where $I$ is the identity matrix, is then added.
The samples $\{y_i(t)\}_{i=1}^n$ of random variable $Y$ are defined as a \ac{hflm} \citep{malfait2003historical} over $t = [0, 1]$ by: 
\begin{equation}\label{eq:hist_data}
    y_i (t) = a \int_0^t x_i(s) \; \beta(s,t) \; ds \; ,
\end{equation}
where $a \in [0, 1]$. 
The function $\beta (s, t)$ maps the dependence from $x_i$ at any point $s$ to $y_i$ at any point $t$, and is defined here as a hyperbolic paraboloid function:
\begin{align}
\beta(s,t) = 8 (s - c_1)^2 - 8 (t - c_2)^2,
\label{eq:beta_fun}
\end{align}
with coefficients $c_1$ and $c_2$ drawn independently from a uniform distribution over the interval  $[0.25,0.75]$.
Afterwards, the samples $y_i(t)$ are evaluated over a regular mesh of 100 points $t \in [0, 1]$ and random noise $\epsilon_i^Y(t) \sim \mathcal{N}(0, I)$ is added. 
Clearly, for $a=0$ our samples are independent, as the samples $y_i(t)=\epsilon_i^Y(t)$ are just random noise. 
As $a$ is increased, the dependence between the samples $x_i(s)$ and $y_i(t)$ becomes easier to detect. 
Figure~\ref{fig:bivariate_indep_results} shows that our test stays within the allowable false-positive rate $\alpha$ and detects the dependence as soon as $a > 0$, even with a low number of samples.

\subsubsection{Joint independence test} \label{sec:joint_experiments}

To produce the synthetic data for joint independence tests, we first generate random \acp{dag} with $d$ nodes using an Erd\"{o}s-R\'enyi model with density 0.5 
\citep[\S~5.2]{pfister2018kernel}. 
For each \ac{dag}, we first use Equation~\eqref{eq:func_data} to produce function-valued random samples for the variables $X^i$ without parents. 
We then generate samples for other variables using a historical functional linear model:  
\begin{equation}\label{eq:multi_hist_data}
    x_i^j (t) = a \sum_{p \in \mathbf{PA}^j} \int_0^t x_i^p(s) \; \beta^p(s,t) \; ds \; ,
\end{equation}
where $\mathbf{PA}^j$ are the parents of node $X^j$, and the function $\beta^p(s,t)$ is given by Equation~\eqref{eq:beta_fun} with random coefficients $c_1$ and $c_2$ independently generated for each descendant-parent pair indexed by $p$.
After being evaluated at a regular mesh of 100 points within $t \in [0, 1]$, random noise $\epsilon_i^j(t) \sim \mathcal{N}(0,I)$ is added.
Note that, again, an increase in the factor $a \in [0,1]$ should make the dependence structure of the \ac{dag} easier to detect. 
Figure \ref{fig:joint_indep_results} shows the test size where $a=0$, resulting in independent variables, and test power where $a>0$. 
We evaluate the joint independence test for $d=4$ variables over various values of $a$ and a range of sample sizes.

\begin{figure}[htb]
    \centering
    \begin{subfigure}{.45\textwidth}
    \includegraphics[width=\linewidth]{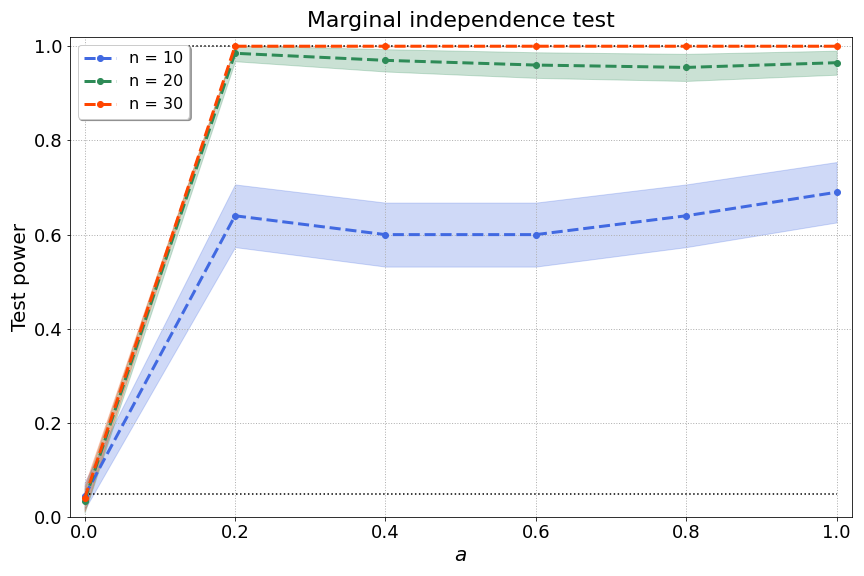}
    \caption{} \label{fig:bivariate_indep_results}
    \end{subfigure}\hspace*{5ex}
    \begin{subfigure}{0.45\textwidth}
    \includegraphics[width=\linewidth]{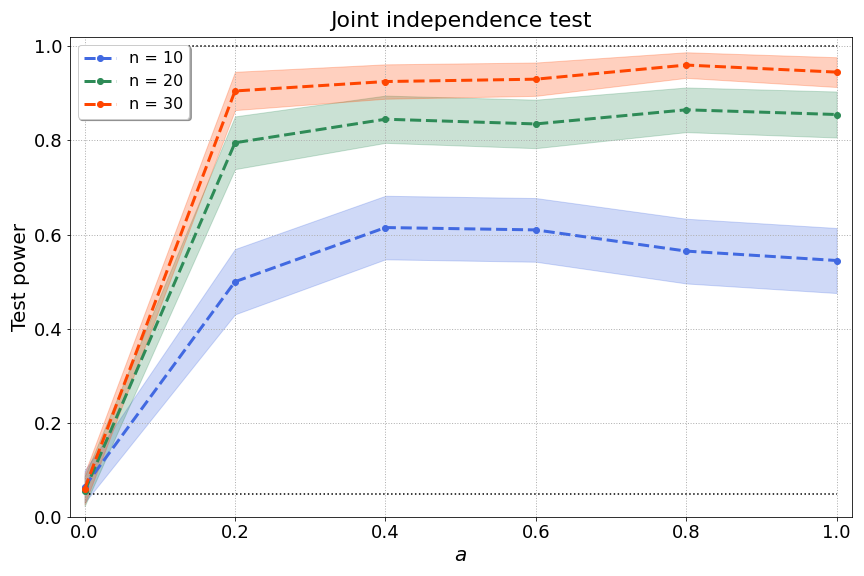}
    \caption{} \label{fig:joint_indep_results}
    \end{subfigure}
    \caption{(a) Bivariate and (b) joint independence tests over various values of the dependence factor $a$ and sample size $n$. The joint independence test is conducted with $d = 4$ variables.}
    \label{fig:bivariate_joint_indep}
\end{figure}

\subsubsection{Conditional independence test} \label{sec:cond_experiments}

To evaluate whether $X$ is independent of $Y$ given $\mathcal{Z}$, where $\mathcal{Z}$ may be any subset within the power set of the $d \setminus \{X, Y\}$ remaining variables, we generate data samples according to:
\begin{equation}
\begin{aligned} \label{eq:data_cond_indep}
    z_i^j (u) &= \sum_{m=1}^M c_{i,m}^j \phi_{m,T} (u), \, j=1,\ldots,|\mathcal{Z}| \\
    x_i (s) &= \sum_{p \in \mathbf{PA}^X} \int_0^s z_i^p(u) \; \beta^p(u,s) \; du \\
    y_i (t) &= \sum_{p \in \mathbf{PA}^Y} \int_0^t z_i^p(u) \; \beta^p(u,t) \; du + a' \int_0^t x_i(s) \; \beta^Y(s,t) \; ds \, , 
\end{aligned}
\end{equation}
 where $|\mathcal{Z}|$ is the cardinality of the set $\mathcal{Z}$, $\beta$'s are given by Equation~\eqref{eq:beta_fun}, and noise terms $\epsilon_i^{Z^j}(u)$, $\epsilon_i^X(s)$ and $\epsilon_i^Y(t)$, are added to the discrete observation values. 

We then apply Algorithm~\ref{alg:search_algorithm} to compute the statistics for our conditional independence test. Firstly, we find that the optimised regularisation strength $\lambda^*$ is robust and reproducible for different random realisations with the same sample size and model parameters. 
Consequently, we search for one optimised $\lambda^*$ (line 1--18 in Algorithm~\ref{alg:search_algorithm}) for each sample size, and fix this value for all 200 independent trials. 
%
Figure~\ref{fig:opt_lambs} summarises the results for $\lambda^*$ for $n \in \{100, 200, 300\}$ samples after conducting a grid search over a range of $L = 18$ possible values $10^{-4} \leq \lambda \leq 10^{-1}$, with $P = 1000$ permutations, and $B = 100$ rejection iterations.
Note that the range of values for $\lambda$ can be tuned in stages by the practitioner, e.g., starting with a coarse initial exploration followed by a more fine-grained range. We recommend to choose $B \geq 100$ for ease of comparison of the evaluation rejection rate to the acceptable false positive rate, $\alpha$.
We find that the optimised $\lambda^*$ exhibits a saturation as the number of conditional variables (the dimension of the conditional set) is increased. 
Given that we perform a kernel ridge regression, a larger number of samples $n$ should result in lower requirements for regularisation---which aligns with our observations of a decreasing $\lambda^*$ over the increase in number of samples $n$. 
Therefore, 
Algorithm~\ref{alg:search_algorithm} optimises the regularisation parameter $\lambda^*$ such that the evaluation rejection rate of the test is closest to the allowable false positive rate $\alpha = 0.05$.

\begin{figure}[htb!]
    \centering
    \includegraphics[width=0.45\linewidth]{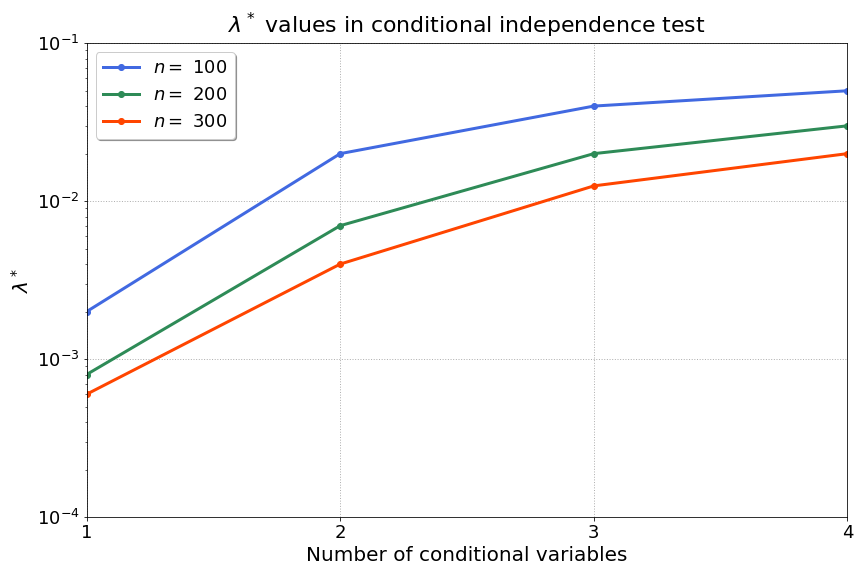}
    \caption{The optimised $\lambda^*$ for increasing dimension $d$ of the evaluated size of the conditional set and different sample sizes $n \in \{100, 200, 300\}$. 
    }
    \label{fig:opt_lambs}
\end{figure}

Figures~\ref{fig:cond_indep_1}--\ref{fig:cond_indep_4} show the results of the conditional independence test for increasing dimension $d$ of the conditional set (i.e., number of conditional variables).
We find that the test power is well preserved as $d$ increases through a concomitant increase of  $\lambda^*$ that partly ameliorates the ``curse of dimensionality'' \citep{shah2020hardness}.
Furthermore, the values of the test power for $a'=0$ 
correspond to the type-I error rates. 

\begin{figure}[htb!]
    \centering
    \begin{subfigure}{0.45\textwidth}
    \includegraphics[width=\linewidth]{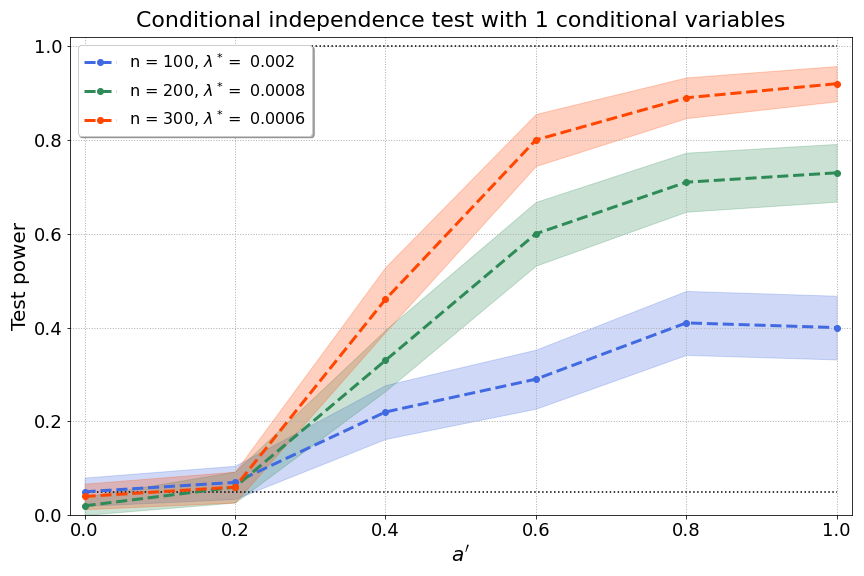}
    \caption{} 
    \label{fig:cond_indep_1}
    \end{subfigure}\hspace*{5ex}
    \begin{subfigure}{0.45\textwidth}
    \includegraphics[width=\linewidth]{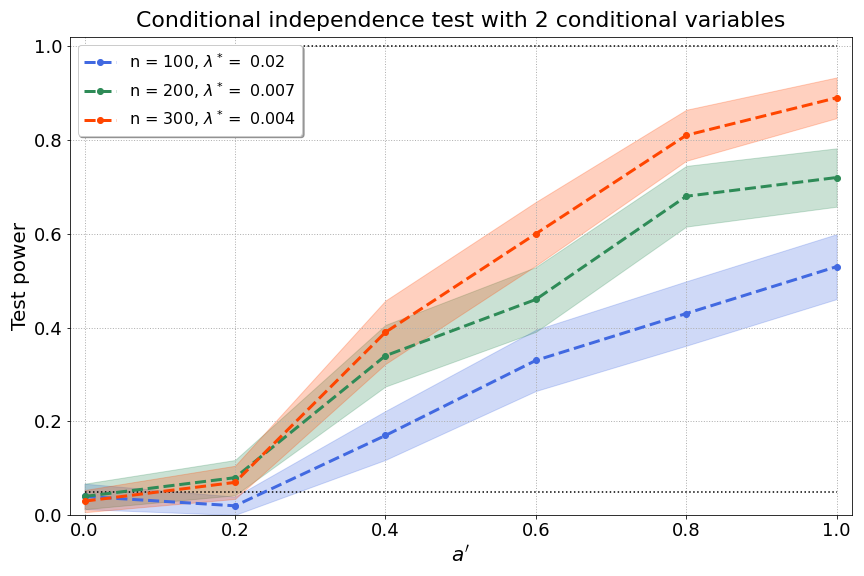}
    \caption{} 
    \label{fig:cond_indep_2}
    \end{subfigure}
    
    \medskip
    \begin{subfigure}{0.45\textwidth}
    \includegraphics[width=\linewidth]{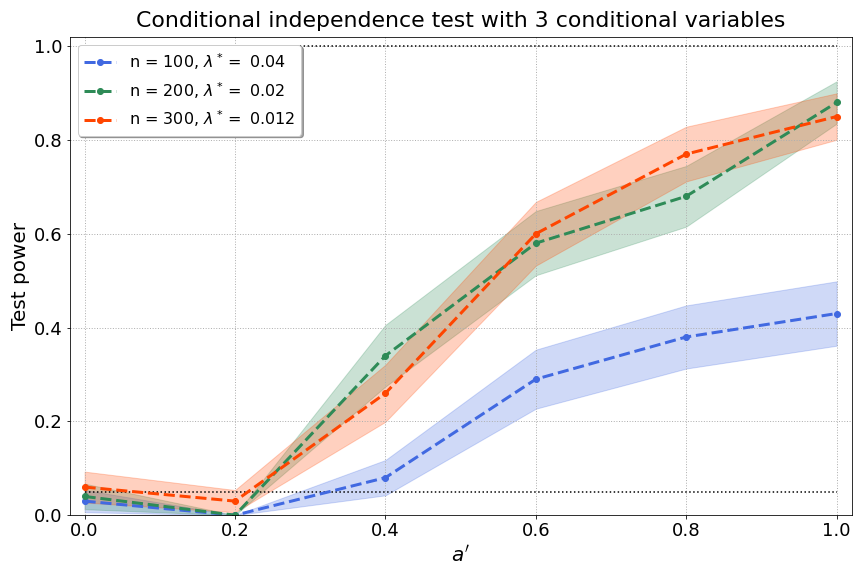}
    \caption{}
    \label{fig:cond_indep_3}
    \end{subfigure}\hspace*{5ex}
    \begin{subfigure}{0.45\textwidth}
    \includegraphics[width=\linewidth]{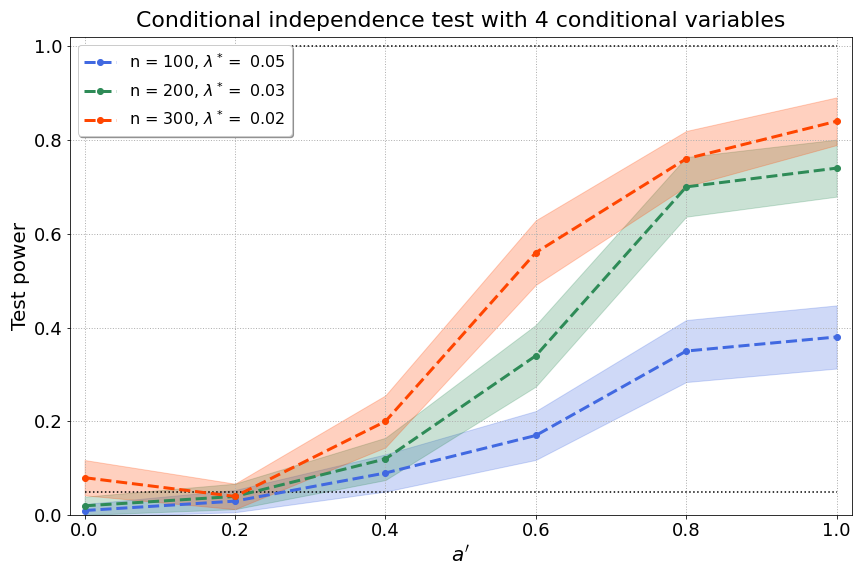}
    \caption{} 
    \label{fig:cond_indep_4}
    \end{subfigure}
    
    \caption{The test power of the conditional independence tests over various values for $a'$ and sample sizes $n$, from (a) 1 to (d) 4 conditional variables. The regularisation parameter $\lambda^*$ is given in the legend.}
    \label{fig:cond_indep}
\end{figure}

\subsection{Causal structure learning}
We use the three independence tests evaluated numerically in Section~\ref{sec:independence_tests} to learn different causal structures among function-valued random variables. 
We start with regression-based methods for the bivariate case $X \rightarrow Y$ and extend to the multivariate case through regression-, constraint-based, and a combination of constraint- and regression-based causal structure learning approaches.

\subsubsection{Synthetic data} \label{sec:causal_experiments}
\emph{Regression-based causal discovery.} \  
We start by evaluating synthetic data for a bivariate system $X \rightarrow Y$ generated according to Equations~\eqref{eq:func_data}--\eqref{eq:hist_data} with $a = 1$. We generate 200 independent trials, and we score the performance of the method using the \ac{shd} (see Appendix~\ref{app:defs_metrics} for its definition).  
As seen in Figure~\ref{fig:regression_2}, our regression-based algorithm using \acs{resit} outperforms two well-known algorithms for causal discovery: Granger-causality, which assumes linearity in the relationships \citep{granger1969investigating}, and \ac{ccm}, which allows for nonlinearity through a time-series embedding \citep{sugihara2012detecting}. 

We then evaluate how the regression-based approach performs in a system of three variables where data are generated according to random \acp{dag} with three nodes that entail historical nonlinear dependence. 
One of the above methods (\ac{ccm}) is commonly applied to bivariate problems only, hence we compare our method on a system with three variables to multivariate Granger causality \citep{geweke1982measurement} as well as to \ac{pcmci}, an algorithm where each time point is represented as a node in a graph~\citep{runge2019detecting} from which we extract a directed graph for the variables. 
Figure~\ref{fig:regression_3} shows substantially improved performance (lower \ac{shd}) of our regression-based algorithm with respect to both of these methods on 3-variable systems. 
Details of all the methods are given in Appendix~\ref{app:comparison_methods}.

\begin{figure}[htb!]
    \centering
    \begin{subfigure}{0.45\textwidth}
        \includegraphics[width=\linewidth]{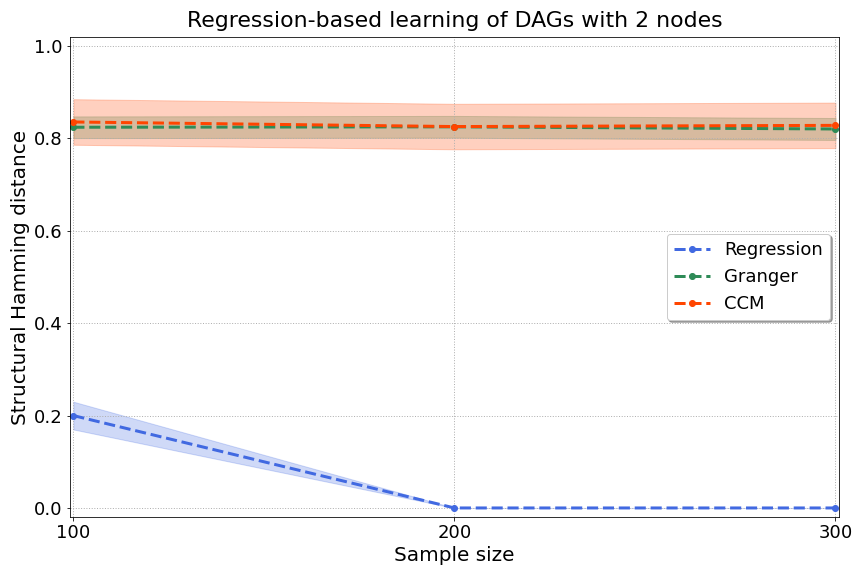}
        \caption{} \label{fig:regression_2}
    \end{subfigure}\hspace*{5ex}
    \begin{subfigure}{0.45\textwidth}
        \includegraphics[width=\linewidth]{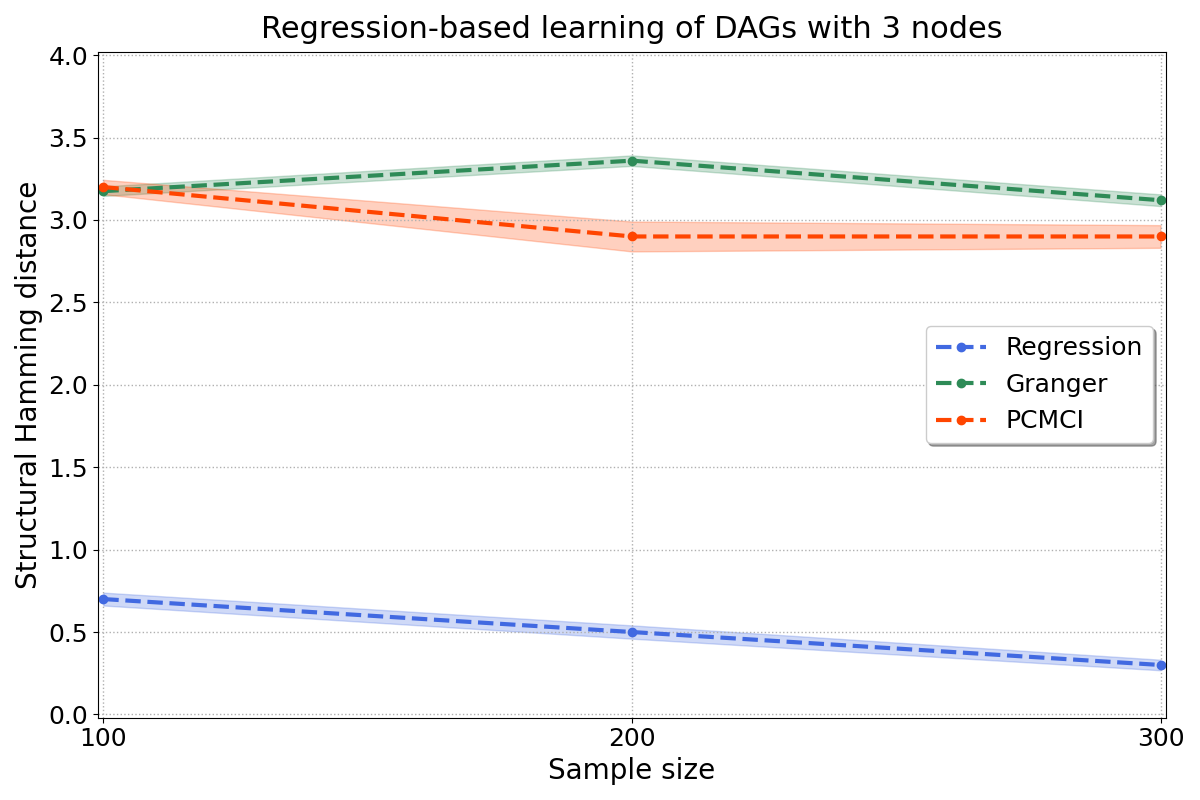}
        \caption{} \label{fig:regression_3}
    \end{subfigure}
    \caption{Accuracy of the regression-based causal discovery using kernel-based joint independence tests among residuals for (a) two and (b) three variables.  In the case of two variables (a), we compare our kernel-based method to Granger-causality and \ac{ccm}.
    In the case of three variables (b), we compare to multivariate Granger-causality} as well as using \ac{pcmci} to produce comparable results.
    Our method significantly outperforms both Granger-causality and \ac{ccm} in the bivariate setting, with just a few mistakes made for low sample numbers ($n=100$) and no mistakes for higher sample sizes.
    For three variables, our method is substantially more accurate than \ac{pcmci} and Granger-causality, with an average \ac{shd} of $0.3$ at $n = 300$ versus $2.9$ for \ac{pcmci} and $3.12$ for Granger-causality.
    The accuracy of the methods is measured in terms of the structural Hamming distance (\ac{shd}).
    \label{fig:shd_regression}
\end{figure}

\emph{Constraint-based causal discovery.} \
The accuracy of our constraint-based approach is evaluated by computing three metrics: normalised \ac{shd} ($\textsc{shd}_{\text{norm}}$), precision, and recall (see Appendix~\ref{app:defs_metrics} for definitions).
Figure~\ref{fig:constraint} shows the accuracy of our algorithm when applied to $d \in \{ 3, 4, 5, 6 \}$ variables.
The normalised \ac{shd} in Figure~\ref{fig:constraint-shd} demonstrates consistent accuracy of the constraint-based causal structure learning method across sample sizes for all $d$. 
To complement this measure, we can examine jointly precision and recall in Figure~\ref{fig:constraint-precision}--\ref{fig:constraint-recall}, where we find that the learnt edges are predominantly unoriented (low values of recall) but the oriented edges that are found are indeed correctly oriented (high values of precision). 

\begin{figure}[htb!]
    \centering
    \begin{subfigure}{0.45\textwidth}
    \includegraphics[width=\linewidth]{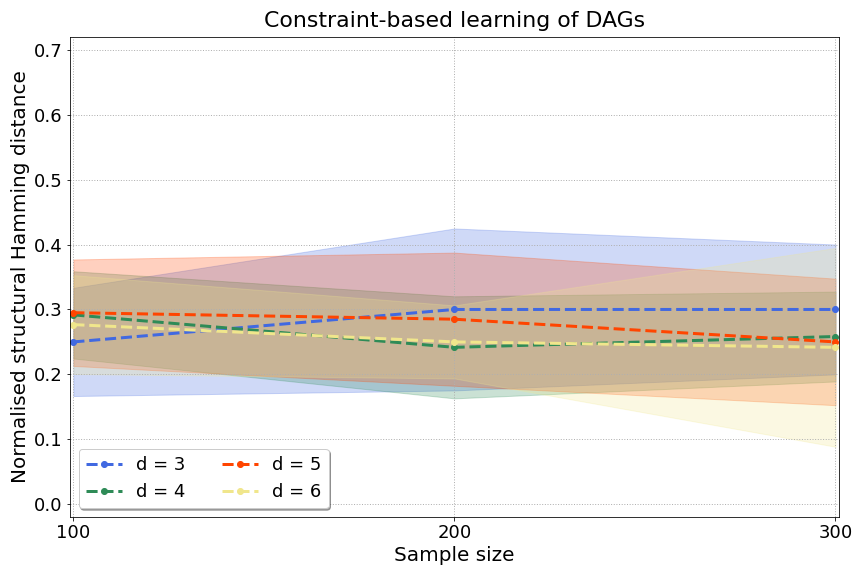}
    \caption{}
    \label{fig:constraint-shd}
    \end{subfigure}\hspace*{\fill}
    \begin{subfigure}{0.45\textwidth}
    \includegraphics[width=\linewidth]{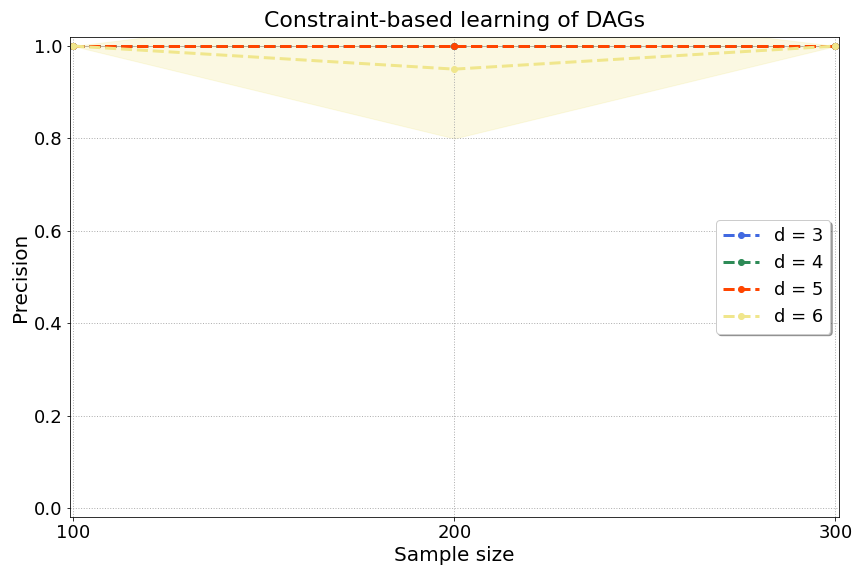}
    \caption{} 
    \label{fig:constraint-precision}
    \end{subfigure}\hspace*{\fill}
    \vspace{2ex}
    \newline
    \begin{subfigure}{0.45\textwidth}
    \includegraphics[width=\linewidth]{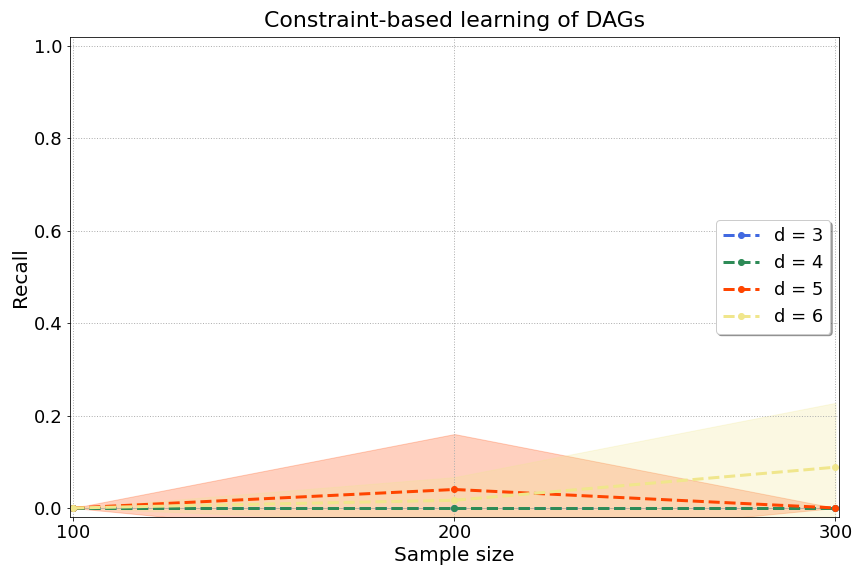}
    \caption{}
    \label{fig:constraint-recall}
    \end{subfigure}
    
    \caption{Constraint-based causal structure learning experiments. From left to right, we compute (a) the normalised \ac{shd}, (b) precision and (c) recall over $d \in \{3, 4, 5, 6 \}$ variables and $n \in \{100, 200, 300\}$ samples with $a = 1$ as the dependence between any two variables in the data that are connected by an edge in the true \ac{dag}.}
    \label{fig:constraint}
\end{figure}

\emph{Combined approach.} \ 
To scale up our causal discovery techniques more efficiently to larger graphs, it is possible to combine constraint- and regression-based causal learning, yielding \acp{dag}.
In this ``combined'' approach we start with a constraint-based step through which we learn \acp{cpdag} by applying conditional independence tests to any two variables conditioned on any subset of the remaining $d-2$ variables 
and orient edges according to the Algorithm in Appendix~\ref{app:alg+meek_rules}, which finds v-structures, applies Meek's orientation rules and returns the Markov equivalence class. 
Often, the Markov equivalence class entails undirected edges, but we can then take a second step using a regression-based approach, under a different set of assumptions, to orient edges by applying \acs{resit}. 
This two-step process yields a \ac{dag}, i.e., every edge in the graph is oriented, in a more scalable manner than applying directly regression-based causal discovery. 
We then measure the accuracy of our approach computing the normalised \ac{shd} \eqref{eq:shd_norm} as above.
Figure~\ref{fig:combined} shows that our results for $d=3,4,5,6$ variables compare favourably to \ac{pcmci} and to multivariate Granger causality with consistently lower \ac{shd} for all $d$.

\begin{figure}[htb!]
    \centering
    \includegraphics[width=0.45\linewidth]{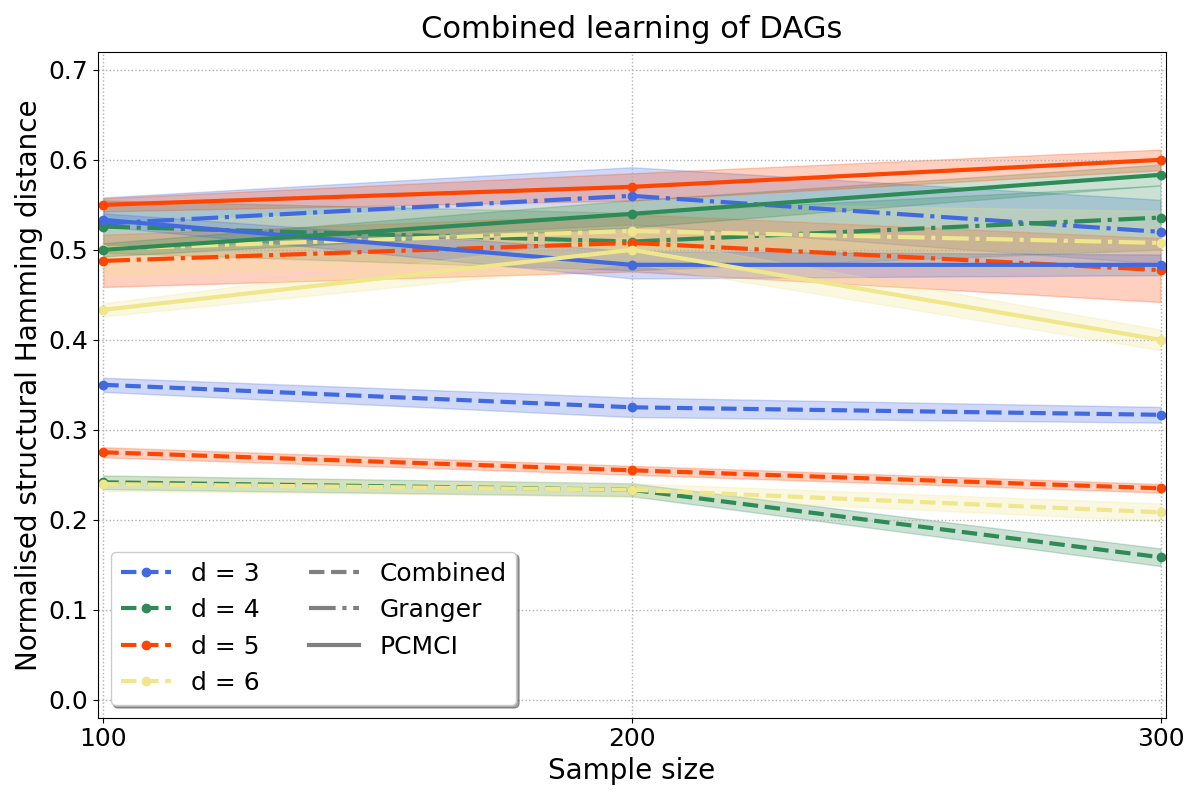}
    
    \caption{Causal structure learning with the ``combined'' approach, where we first apply the constraint-based method to find the Markov equivalence class, followed by the regression-based method to orient the undirected edges of the Markov equivalence class.
    We compute the normalised \ac{shd} over $d \in \{3, 4, 5, 6\}$ variables and $n \in \{100, 200, 300\}$ samples with $a = 1$ as in Equation~\eqref{eq:multi_hist_data}. 
    Our results \textit{(dashed lines)} have substantially lower  normalised \ac{shd} than those obtained from Granger-causality \textit{(dash-dotted lines)} and} \ac{pcmci} \textit{(solid lines)}, applied as described in Appendix~\ref{app:comparison_methods}.
    \label{fig:combined}
\end{figure}

\subsubsection{Real-world data} \label{sec:wgi_data}

To showcase the application of our methods to real data, we return now to the motivating example on socio-economic indicators mentioned in the Introduction (Section~\ref{sec:introduction}).
\citet{sulemana2018empirical} tested for the causal direction between \textit{corruption} and \textit{income inequality}, as measured by the Control of Corruption index \citep{kaufmann2011worldwide} and Gini coefficient \citep{gini1936measure}, respectively.  
They applied Granger-causality with Wald-$\chi^2$ tests, which statistically test the null hypothesis that an assumed explanatory variable does not have a (linear) influence on other considered variables.
Each variable, \textit{corruption} and \textit{income inequality}, was tested against being the explanatory variable (i.e., the cause) in their bivariate relationship. 
Their findings showed that \textit{income inequality} as the explanatory variable results in a p-value of $0.079$, whereas \textit{corruption} is more likely to be the cause with a p-value of $0.276$.
Analysing the same data using \ac{ccm}, which does not assume linearity, we also find the same cause-effect direction:  $70.9\%$ of the sample pairs $(x_i, y_i)$ validate corruption as the cause, against $29.7\%$ of pairs which detect the opposite direction.
However, using our proposed regression-based kernel approach, which does not assume linearity or stationarity,  we find the opposite result, i.e., the causal dependence of \textit{corruption} on \textit{income inequality} is statistically more likely (p $= 0.0859$) than the reverse direction (p $= 0.0130$) when applying \ac{resit} that rejects the null hypothesis of independence of the residuals after regression in the false causal directions. 

Going beyond pairwise interactions, we have also applied our causal structure learning method to the set of World Governance Indicators (\acp{wgi})~\citep{kaufmann2011worldwide}. 
The \acp{wgi} are a collection of six variables, which have been measured in 182 countries over 25 years from 1996 to 2020 (see Appendix~\ref{app:WGI}). 
We view the time-series of the countries as independent functional samples of each variable (so that $n=182$ and the functional dimension is time), and we apply our methods as described above to establish the causal structure amongst the six \acp{wgi}.  
In this case, we use the ``combined'' approach by successively applying our constraint- and regression-based causal structure learning methods. 
We first learn an undirected causal skeleton from the \acp{wgi} (Figure~\ref{fig:WGI_undirected}), and we find a triangle between Government Effectiveness (GE), Rule of Law (RL) and Control of Corruption (CC), and a separate link between Voice and Accountability (VA) and Regulatory Quality (RQ). 
We then orient these edges (Figure~\ref{fig:WGI_directed}) and find that Government Effectiveness (GE) causes both Rule of Law (RL) and Control of Corruption (CC), and Rule of Law (RL) causes Control of Corruption (CC). 
We also find that Voice and Accountability (VA) causes Regulatory Quality (RQ).

\begin{figure}[htb]
    \centering
    \begin{subfigure}{0.45\textwidth}
    \includegraphics[width=\linewidth]{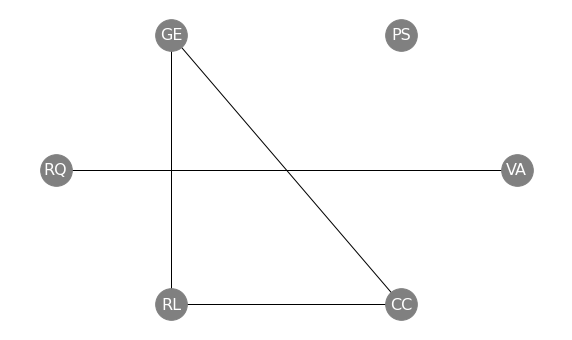}
    \caption{}\label{fig:WGI_undirected}
    \end{subfigure}\hspace*{5ex}
    \begin{subfigure}{0.45\textwidth}
    \includegraphics[width=\linewidth]{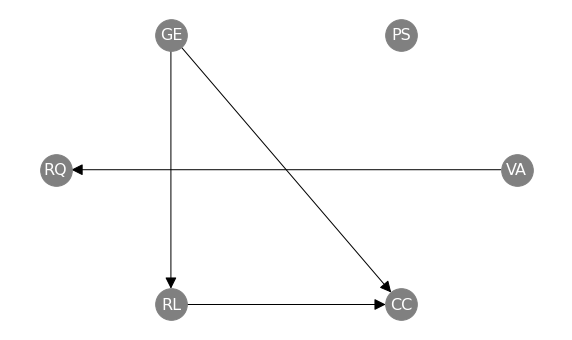}
    \caption{}\label{fig:WGI_directed}
    \end{subfigure}
    
    \caption{Undirected (a) and directed (b) causal networks on the World Governance Indicators data set to evaluate the results of the constraint-based approach alone (a) and the subsequent regression-based approach (b). 
    The labels are abbreviations of the official names: Voice and Accountability (VA), 
    Political Stability (PS), 
    Government Effectiveness (GE),
    Regulatory Quality (RQ), 
    Rule of Law (RL), and
    Control of Corruption (CC) (see Appendix~\ref{app:WGI}).}
    \label{fig:WGI_data}
\end{figure}

\section{Discussion and Conclusion}
\label{sec:discussion}

We present a causal structure learning framework on functional data that utilises kernel-based independence tests to extend the applicability of the widely used regression- and constraint-based approaches. 
The foundation of the framework originates from the functional data analysis literature by interpreting any discrete-measured observations of a random variable as finite-dimensional realisations of functions.

Using synthetic data, we have demonstrated that our regression-based approach outperforms existing methods such as Granger-causality, \ac{ccm} and \ac{pcmci} when learning causal relationships between two and three variables.
In the bivariate case, we have carried out a more detailed comparison to Granger-causality and \ac{ccm} to explore the robustness of our regression-based approach to nonlinearity and nonstationarity in the data (see Appendix~\ref{app:linear_stationary}) We find that
while Granger degrades under the introduction of nonlinearity in the data, and \ac{ccm} degrades under the introduction of nonstationarity, our method remains robust in its performance under both nonlinearity and nonstationarity. 
In addition, as seen in Figure~\ref{fig:regression_2}, \ac{ccm} can have difficulty in detecting strong unidirectional causal dependencies~\citep{sugihara2012detecting,ye2015distinguishing},
where the cause variable $X$  uniquely determines the state of the effect variable $Y$ inducing ``generalised synchrony'' \citep{rulkov1995generalized}. In such cases,  \ac{ccm} can predict samples of $Y$ from $X$ equally well as $X$ from $Y$; hence \ac{ccm} finds $X \rightarrow Y$ and $X \leftarrow Y$ indistinctly. In contrast, our experiments (Section~\ref{sec:causal_experiments} and Appendix~\ref{app:linear_stationary}) show that our regression-based method is unaffected by the presence of ``generalised synchrony'' in the data.

Further, we show that our conditional independence test, which is the cornerstone of the constraint-based causal discovery approach, achieves type-I error rates close to the acceptable false-positive rate $\alpha$ and high type-II error rates, even when the number of variables in the conditional set increases. 
\citet{shah2020hardness} rightly state that any conditional independence test can suffer from an undesirable test size or low test power in finite samples, and our method is obviously no exception.
However, Figure~\ref{fig:cond_indep} demonstrates the counterbalance of the ``curse of dimensionality'' through the optimised regularisation strength $\lambda^*$. 
Indeed, while with larger numbers of conditional variables, we would generally expect the test power to diminish, the optimisation of $\lambda^*$ offsets this reduction, 
resulting in no significant decrease in test power with a growing number of conditional variables.
Although our proposed method is computationally expensive and significantly benefits from high sample sizes, a suitable regularisation strength could be chosen for large conditional sets in principle.

Moreover, we demonstrate that constraint- and regression-based causal discovery methods can be combined to learn \acp{dag} with large number of nodes, which would otherwise be computationally very expensive when relying on regression-based methods to yield \acp{dag}.
By comparing Figure~\ref{fig:constraint-shd} and \ref{fig:combined} we see however that it does not necessarily result in a lower \ac{shd}. 
After learning the Markov equivalence class through the constraint-based approach, we apply \ac{resit} to orient undirected edges in the ``combined'' approach. 
Here, mistakes in the orientation of undirected edges add 2 to the \ac{shd} whereas an undirected edge only adds 1 to the \ac{shd}.
When applied to real-world data, we have utilised this approach to learn a causal graph of the \acp{wgi}, assuming each relationship between two variables is uni-directionally identifiable as suggested by several economic studies \citep{jong2005comparative,alesina2005corruption,dobson2010there}.

The presented work contributes to opening the field of causal discovery to functional data.
Beyond the results presented here, we believe that more research needs to be conducted to (i) increase the efficiency of the independence tests, meaning smaller sample sets can achieve higher test power; (ii) learn about upper bounds of the regularisation strength $\lambda^*$ with respect to the size of the sample and conditional sets; (iii) reduce the computational cost of the conditional independence test (see Appendices~\ref{app:complexity}~and~\ref{app:running_time} for a comparison to other methods); and (iv) establish connections and investigate differences to causal structure learning approaches based on transfer entropy \citep{barnett2009granger,porta2015conditional}.

\clearpage
\bibliography{references}

\newpage
\appendix
\appendix





\section{Proof of Theorem~\ref{th:hscic_consistent}}\label{app:proof}

\begin{proof}
		All norms \(\lVert\cdot\rVert\) and inner products \(\langle\cdot,\cdot\rangle\) in this proof are taken with respect to the RKHS in which the normand and the arguments of the inner product reside in. 
		
		Denote by \(\text{HSCIC}(\mathbf{X},\mathbf{Y},\mathbf{Z})\) the \textit{true} Hilbert-Schmidt conditional independence criterion between the random variables \(\mathbf{X}\) and \(\mathbf{Y}\) given \(\mathbf{Z}\), i.e. 
		\[\text{HSCIC}(\mathbf{X},\mathbf{Y},\mathbf{Z})=\lVert\mu_{\mathbf{X}\mathbf{Y}|\mathbf{Z}}-\mu_{\mathbf{X}|\mathbf{Z}}\otimes\mu_{\mathbf{Y}|\mathbf{Z}}\rVert,\]
		and denote by \(\hat{\text{HSCIC}}(\mathbf{X},\mathbf{Y},\mathbf{Z},\lambda_n)\) the empirical estimate of \(\text{HSCIC}(\mathbf{X},\mathbf{Y},\mathbf{Z})\) based on samples, i.e.
		\[\hat{\text{HSCIC}}(\mathbf{X},\mathbf{Y},\mathbf{Z},\lambda_n)=\lVert\hat{\mu}_{\mathbf{X}\mathbf{Y}|\mathbf{Z}}-\hat{\mu}_{\mathbf{X}|\mathbf{Z}}\otimes\hat{\mu}_{\mathbf{Y}|\mathbf{Z}}\rVert,\]
		where \(\hat{\mu}_{\mathbf{X}\mathbf{Y}|\mathbf{Z}}\), \(\hat{\mu}_{\mathbf{X}|\mathbf{Z}}\) and \(\hat{\mu}_{\mathbf{Y}|\mathbf{Z}}\) are the empirical estimates of the conditional mean embeddings \(\mu_{\mathbf{X}\mathbf{Y}|\mathbf{Z}}\), \(\mu_{\mathbf{X}|\mathbf{Z}}\) and \(\mu_{\mathbf{Y}|\mathbf{Z}}\) obtained with the regularisation parameter \(\lambda_n\). 
		
		Note that by the reverse triangle inequality, followed by the ordinary triangle inequality, we have
		\begin{alignat*}{2}
			\left\lvert\text{HSCIC}(\mathbf{X},\mathbf{Y},\mathbf{Z})-\hat{\text{HSCIC}}(\mathbf{X},\mathbf{Y},\mathbf{Z},\lambda_n)\right\rvert&=\left\lvert\lVert\mu_{\mathbf{X}\mathbf{Y}|\mathbf{Z}}-\mu_{\mathbf{X}|\mathbf{Z}}\otimes\mu_{\mathbf{Y}|\mathbf{Z}}\rVert-\lVert\hat{\mu}_{\mathbf{X}\mathbf{Y}|\mathbf{Z}}-\hat{\mu}_{\mathbf{X}|\mathbf{Z}}\otimes\hat{\mu}_{\mathbf{Y}|\mathbf{Z}}\rVert\right\rvert\\
			&\leq\lVert\mu_{\mathbf{X}\mathbf{Y}|\mathbf{Z}}-\mu_{\mathbf{X}|\mathbf{Z}}\otimes\mu_{\mathbf{Y}|\mathbf{Z}}-\hat{\mu}_{\mathbf{X}\mathbf{Y}|\mathbf{Z}}+\hat{\mu}_{\mathbf{X}|\mathbf{Z}}\otimes\hat{\mu}_{\mathbf{Y}|\mathbf{Z}}\rVert\\
			&\leq\lVert\mu_{\mathbf{X}\mathbf{Y}|\mathbf{Z}}-\hat{\mu}_{\mathbf{X}\mathbf{Y}|\mathbf{Z}}\rVert+\lVert\mu_{\mathbf{X}|\mathbf{Z}}\otimes\mu_{\mathbf{Y}|\mathbf{Z}}-\hat{\mu}_{\mathbf{X}|\mathbf{Z}}\otimes\hat{\mu}_{\mathbf{Y}|\mathbf{Z}}\rVert.
		\end{alignat*}
		Here,
		\begin{alignat*}{2}
			\lVert\mu_{\mathbf{X}|\mathbf{Z}}\otimes\mu_{\mathbf{Y}|\mathbf{Z}}-\hat{\mu}_{\mathbf{X}|\mathbf{Z}}\otimes\hat{\mu}_{\mathbf{Y}|\mathbf{Z}}\rVert^2&=\langle\mu_{\mathbf{X}|\mathbf{Z}}\otimes\mu_{\mathbf{Y}|\mathbf{Z}},\mu_{\mathbf{X}|\mathbf{Z}}\otimes\mu_{\mathbf{Y}|\mathbf{Z}}\rangle\\
			&\quad-2\langle\mu_{\mathbf{X}|\mathbf{Z}}\otimes\mu_{\mathbf{Y}|\mathbf{Z}},\hat{\mu}_{\mathbf{X}|\mathbf{Z}}\otimes\hat{\mu}_{\mathbf{Y}|\mathbf{Z}}\rangle\\
			&\qquad+\langle\hat{\mu}_{\mathbf{X}|\mathbf{Z}}\otimes\hat{\mu}_{\mathbf{Y}|\mathbf{Z}},\hat{\mu}_{\mathbf{X}|\mathbf{Z}}\otimes\hat{\mu}_{\mathbf{Y}|\mathbf{Z}}\rangle\\
			&=\langle\mu_{\mathbf{X}|\mathbf{Z}},\mu_{\mathbf{X}|\mathbf{Z}}\rangle\langle\mu_{\mathbf{Y}|\mathbf{Z}},\mu_{\mathbf{Y}|\mathbf{Z}}\rangle\\
			&\quad-2\langle\mu_{\mathbf{X}|\mathbf{Z}},\hat{\mu}_{\mathbf{X}|\mathbf{Z}}\rangle\langle\mu_{\mathbf{Y}|\mathbf{Z}},\hat{\mu}_{\mathbf{Y}|\mathbf{Z}}\rangle\\
			&\qquad+\langle\hat{\mu}_{\mathbf{X}|\mathbf{Z}},\hat{\mu}_{\mathbf{X}|\mathbf{Z}}\rangle\langle\hat{\mu}_{\mathbf{Y}|\mathbf{Z}},\hat{\mu}_{\mathbf{Y}|\mathbf{Z}}\rangle\\
			&=\langle\mu_{\mathbf{X}|\mathbf{Z}},\mu_{\mathbf{X}|\mathbf{Z}}\rangle\langle\mu_{\mathbf{Y}|\mathbf{Z}},\mu_{\mathbf{Y}|\mathbf{Z}}\rangle\\
			&\quad-2\langle\mu_{\mathbf{X}|\mathbf{Z}},\mu_{\mathbf{X}|\mathbf{Z}}\rangle\langle\mu_{\mathbf{Y}|\mathbf{Z}},\hat{\mu}_{\mathbf{Y}|\mathbf{Z}}\rangle+2\langle\mu_{\mathbf{X}|\mathbf{Z}},\mu_{\mathbf{X}|\mathbf{Z}}-\hat{\mu}_{\mathbf{X}|\mathbf{Z}}\rangle\langle\mu_{\mathbf{Y}|\mathbf{Z}},\hat{\mu}_{\mathbf{Y}|\mathbf{Z}}\rangle\\
			&\qquad+\langle\mu_{\mathbf{X}|\mathbf{Z}},\mu_{\mathbf{X}|\mathbf{Z}}\rangle\langle\hat{\mu}_{\mathbf{Y}|\mathbf{Z}},\hat{\mu}_{\mathbf{Y}|\mathbf{Z}}\rangle \\
            &\qquad+\langle\hat{\mu}_{\mathbf{X}|\mathbf{Z}}-\mu_{\mathbf{X}|\mathbf{Z}},\hat{\mu}_{\mathbf{X}|\mathbf{Z}}+\mu_{\mathbf{X}|\mathbf{Z}}\rangle\langle\hat{\mu}_{\mathbf{Y}|\mathbf{Z}},\hat{\mu}_{\mathbf{Y}|\mathbf{Z}}\rangle\\
			&=\langle\mu_{\mathbf{X}|\mathbf{Z}},\mu_{\mathbf{X}|\mathbf{Z}}\rangle\lVert\mu_{\mathbf{Y}|\mathbf{Z}}-\hat{\mu}_{\mathbf{Y}|\mathbf{Z}}\rVert^2+2\langle\mu_{\mathbf{X}|\mathbf{Z}},\mu_{\mathbf{X}|\mathbf{Z}}-\hat{\mu}_{\mathbf{X}|\mathbf{Z}}\rangle\langle\mu_{\mathbf{Y}|\mathbf{Z}},\hat{\mu}_{\mathbf{Y}|\mathbf{Z}}\rangle\\
			&\qquad+\langle\hat{\mu}_{\mathbf{X}|\mathbf{Z}}-\mu_{\mathbf{X}|\mathbf{Z}},\hat{\mu}_{\mathbf{X}|\mathbf{Z}}+\mu_{\mathbf{X}|\mathbf{Z}}\rangle\langle\hat{\mu}_{\mathbf{Y}|\mathbf{Z}},\hat{\mu}_{\mathbf{Y}|\mathbf{Z}}\rangle\\
			&\leq\lVert\mu_{\mathbf{X}|\mathbf{Z}}\rVert^2\lVert\mu_{\mathbf{Y}|\mathbf{Z}}-\hat{\mu}_{\mathbf{Y}|\mathbf{Z}}\rVert^2\\
			&\quad+\lVert\mu_{\mathbf{X}|\mathbf{Z}}-\hat{\mu}_{\mathbf{X}|\mathbf{Z}}\rVert\left(2\lVert\mu_{\mathbf{X}|\mathbf{Z}}\rVert\langle\mu_{\mathbf{Y}|\mathbf{Z}},\hat{\mu}_{\mathbf{Y}|\mathbf{Z}}\rangle+\lVert\hat{\mu}_{\mathbf{X}|\mathbf{Z}}+\mu_{\mathbf{X}|\mathbf{Z}}\rVert\lVert\hat{\mu}_{\mathbf{Y}|\mathbf{Z}}\rVert\right),
		\end{alignat*}
		where we used the Cauchy-Schwarz inequality in the last inequality. We note that \citet[Theorem 4.4]{park2020measure} ensures that each of \(\lVert\mu_{\mathbf{X}\mathbf{Y}|\mathbf{Z}}-\hat{\mu}_{\mathbf{X}\mathbf{Y}|\mathbf{Z}}\rVert\), \(\lVert\mu_{\mathbf{Y}|\mathbf{Z}}-\hat{\mu}_{\mathbf{Y}|\mathbf{Z}}\rVert\) and \(\lVert\mu_{\mathbf{X}|\mathbf{Z}}-\hat{\mu}_{\mathbf{X}|\mathbf{Z}}\rVert\) converge to 0 in probability, hence \(\hat{\text{HSCIC}}(\mathbf{X},\mathbf{Y},\mathbf{Z},\lambda_n)\rightarrow\text{HSCIC}(\mathbf{X},\mathbf{Y},\mathbf{Z})\) in probability with respect to the \(L_2\)-norm.
		
		Now we show that the empirical average of the empirical HSCIC converges to the population integral of the population HSCIC. Writing \(\hat{\text{HSCIC}}\) as a shorthand for \(\hat{\text{HSCIC}}(\mathbf{X},\mathbf{Y},\mathbf{Z},\lambda_n)\), we see that
		\begin{alignat*}{2}
			\left\lvert\frac{1}{n}\sum^n_{i=1}\hat{\text{HSCIC}}(z_i)-\int\text{HSCIC}(z)d\mathbb{P}(z)\right\rvert&\leq\left\lvert\frac{1}{n}\sum^n_{i=1}\hat{\text{HSCIC}}(z_i)-\int\hat{\text{HSCIC}}(z)d\mathbb{P}(z)\right\rvert\\
			&\qquad+\left\lvert\int\hat{\text{HSCIC}}(z)d\mathbb{P}(z)-\int\text{HSCIC}(z)d\mathbb{P}(z)\right\rvert\\
			&\leq\sup_{f\in\mathcal{F}}\left\lvert\frac{1}{n}\sum^n_{i=1}f(z_i)-\int f(z)d\mathbb{P}(z)\right\rvert\\
			&\qquad+\left\lvert\int\hat{\text{HSCIC}}(z)d\mathbb{P}(z)-\int\text{HSCIC}(z)d\mathbb{P}(z)\right\rvert.
		\end{alignat*} 
		Here, the first term converges to 0 thanks to the uniform law of large numbers over the function class \(\mathcal{F}\) in which the empirical estimation of \(\hat{\text{HSCIC}}\) is carried out, and the second term converges to 0 as shown above. 
		
		Now, by \citet[Theorem 5.4]{park2020measure}, since our kernel \(k_1\otimes k_2\) is characteristic, we have that \(\text{HSCIC}(\mathbf{X},\mathbf{Y},\mathbf{Z})\) is the identically zero function if and only if \(\mathbf{X}\perp\mathbf{Y}\mid\mathbf{Z}\). Hence, if the alternative hypothesis \(H_1\) holds, we have
		\[\int\text{HSCIC}(z)d\mathbb{P}(z)>0.\]
		See that, for any \(r>0\), we have
		\begin{alignat*}{2}
			\lim_{n\rightarrow\infty}\mathbb{P}\left(\frac{1}{\sqrt{n}}\sum^n_{i=1}\hat{\text{HSCIC}}(z_i)\geq r\right)&=1-\lim_{n\rightarrow\infty}\mathbb{P}\left(\frac{1}{\sqrt{n}}\sum^n_{i=1}\hat{\text{HSCIC}}(z_i)<r\right)\\
			&=1-\lim_{n\rightarrow\infty}\mathbb{P}\left(\frac{1}{n}\sum^n_{i=1}\hat{\text{HSCIC}}(z_i)-\frac{r}{\sqrt{n}}<0\right)\\
			&=1-\lim_{n\rightarrow\infty}\mathbb{P}\left(\int\text{HSCIC}(z)d\mathbb{P}(z)<0\right)\\
			&=1,
		\end{alignat*}
		where we used Portmanteau's theorem \citep[p.6, Lemma 2.2]{van2000asymptotic} to get from the second line to the third, using the convergence of \(\frac{1}{n}\sum^n_{i=1}\hat{\text{HSCIC}}(z_i)\) to \(\int\text{HSCIC}(z)d\mathbb{P}(z)\) shown above. Hence, our test is consistent. 
\end{proof}

\section{Additional definitions for causal learning} \label{app:defs}
\subsection*{d-separation} 
If $X, Y$ and $Z$ are three pairwise disjoint (sets of) nodes in a \ac{dag} $G$, then $Z$ d-separates $X$ from $Y$ in $G$ if $Z$ blocks every possible path between (any node in) $X$ and (any node in) $Y$ in $G$ \citep[\S~1.2.3]{pearl2009causality}. 
We then write $X \indep_{G} Y| Z$. The symbol $\indep_{G}$ is commonly used to denote d-separation.

\subsection*{Faithfulness assumption} 
If two random variables are (conditionally) independent in the observed distribution $\P$, then they are d-separated in the underlying causal graph $G$ \citep[Definition 6.33]{peters2017elements}.

\subsection*{Causal Markov assumption} 
The distribution $\P_{XYZ}$ is Markov with respect to a \ac{dag} $G$ if $X \indep_{G} Y | Z \implies X \indep Y | Z$ for all disjoint (sets of) nodes $X, Y, Z$, where $\indep_{G}$ denotes d-separation \citep[Definition 6.21]{peters2017elements}.

\section{Meek's orientation rules, the SGS and the PC algorithms}\label{app:alg+meek_rules}
In Figure~\ref{fig:meek_rules}, a node is placed at each end of the edges. 
Rule 1 states that if one directed edge (here the vertical edge) points towards another undirected edge (the horizontal edge) that is not further connected to any other edges, the undirected edge must point in the same direction as the adjacent edge, because the three vertices would otherwise form a v-structure (which are only formed with conditional dependence).
Rule 2, in contrast, points the undirected edge in both adjacent vertices' direction because it would violate the acyclicity of a \ac{dag} otherwise. 
Rule 3 and 4 avoid new v-structures which would come into existence by applying Rule 2 twice if the oriented edges pointed in the opposite direction. These new v-structures would be in the left top corners of the rectangles. 

\begin{figure}[htb]
    \centering
    \includegraphics[width=.9\textwidth]{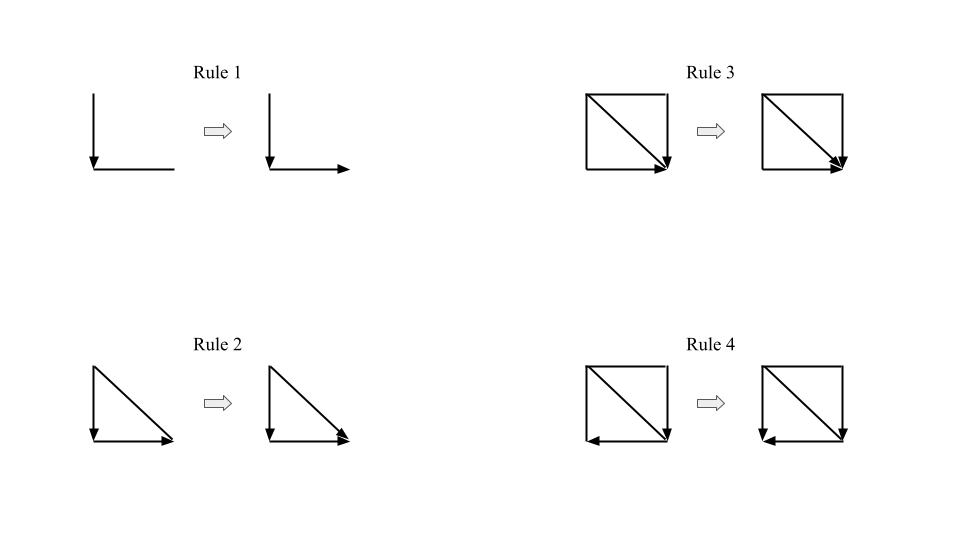}
    \caption{Meek rules to orient edges that remain in the graph after conditional independence tests and edges are oriented based on detected colliders. Proofs that edges cannot be oriented in the opposite direction without violating the acyclicity of the graph and found conditional independencies are given in \citet{meek1995complete}.}
    \label{fig:meek_rules}
\end{figure}

Based on Meek's orientation rules, the SGS algorithm tests each pair of variables as follows \citep[\S~5.4.1]{spirtes2000causation}:

\begin{outline}[enumerate]
 \1 Form the complete undirected graph $G$ from the set of vertices (or nodes) $\mathbf{V}$.
 \1 For each pair of vertices $X$ and $Y$, if there exists a subset $S$ of $\mathbf{V} \setminus \{X, Y\}$ such that $X$ and $Y$ are d-separated, i.e., conditionally independent (causal Markov condition), given $Z$, remove the edge between $X$ and $Y$ from $G$.
 \1 Let $K$ be the undirected graph resulting from the previous step 2. For each triple of vertices $X$, $Y$, and $Z$ such that the pair $X$ and $Y$ and the pair $Y$ and $Z$ are each adjacent in $K$ (written as $X - Y - Z$) but the pair $X$ and $Z$ are not adjacent in $K$, orient $X - Y - Z$ as $X \rightarrow Y \leftarrow Z$ if and only if there is no subset $S$ of $Y$ that d-separates $X$ and $Z$.
 \1 Repeat the following steps until no more edges can be oriented:
   \2 If $X \rightarrow Y$, $Y$ and $Z$ are adjacent, $X$ and $Z$ are not adjacent, and there is no arrowhead of other vertices at $Y$, then orient $Y - Z$ as $Y \rightarrow Z$ (Rule 1 of Meek's orientation rules).
   \2 If there is a directed path over some other vertices from $X$ to $Y$, and an edge between $X$ and $Y$, then orient $X - Y$ as $X \rightarrow Y$ (Rule 2 of Meek's orientation rules).
\end{outline}

The large computational cost of applying the SGS algorithm (especially step 2) to data, due to the large number of combinations of variables as the potential conditional sets, opened the door for computationally more efficient methods. 
The PC algorithm \citep[\S~5.4.2]{spirtes2000causation} is amongst the most popular alternatives and minimises the computational cost by searching for conditional independence in a structured manner, as opposed to step 2 of the SGS algorithm that iterates over all possible conditional subsets to any two variables $A$ and $B$. 
Starting again with a fully connected undirected graph, the PC algorithm begins by testing for marginal independence between any two variables $A$ and $B$ and deletes the edge connecting $A$ and $B$, $A - B$, if $A \indep B$. 
Then, it proceeds with testing for $A \indep B | C$ and erases $A - B$ if one conditional variable $C$ is found that makes $A$ and $B$ independent given $C$. 
Afterwards, the conditional set is extended to two variables, $\{C, D\}$, and the edge $A - B$ is deleted if this conditional set makes $A$ and $B$ independent. 
The conditional set is extended, round after round, until no more conditional independencies are found, resulting in the sparsified graph $K$. Step 3 and 4 are then pursued as in the SGS algorithm.

\section{Definitions of performance metrics} \label{app:defs_metrics}

\subsection*{\ac{shd}} For two graphs, $G_1$ and $G_2$, \ac{shd} is the sum of the element-wise absolute differences between the adjacency matrices $A_1$ and $A_2$ of $G_1$ and $G_2$, respectively:
\begin{equation} \label{eq:shd}
    \ac{shd} = \sum_{ij} | A_1^{ij} - A_2^{ij} | \; .
\end{equation}
Thus, when a learnt causal graph $G_1$ includes an edge $X \rightarrow Y$ oriented opposite to the corresponding edge in the true graph $G_2$  (i.e., $X \leftarrow Y$ is true),  we add 2 to the \ac{shd} score (``double penalty''). 
Note that others \citep[e.g.,][]{peters2015structural} only penalise by 1 for a learnt edge being directed opposite to the true edge.

\subsection*{Normalised \ac{shd}} We also define the normalised \ac{shd}, where we divide \ac{shd} by the number of possible directed edges in a graph of size $d$, thus allowing comparison across graphs with different number of nodes:
\begin{equation}\label{eq:shd_norm}
\textsc{shd}_{\text{norm}} = \frac{\textsc{shd}}{d (d-1)} \; .
\end{equation}

\subsection*{Precision} Let the total number of edges in a learnt graph $G_1$ be the sum of truly and falsely oriented edges, $e = e_t + e_f$ (the true-positives and false-positives, respectively).
The truly oriented edges are the edges that correspond to the zero elements in the matrix containing the element-wise differences between the learnt $A_1$ and the true $A_2$. 
In contrast, the falsely oriented edges are the elements having a value of 2 (``double penalty'').
Precision is the proportion of correctly oriented edges $e_t$ in a learnt graph $G_1$ in comparison to a true graph $G_2$ out of all learnt edges $e$ in $G_1$:
\begin{equation}\label{eq:precision}
    \text{Precision} = \frac{e_t}{e} \; .
\end{equation}

\subsection*{Recall} We separate between oriented and unoriented edges, $e_o$ and $e_u$, respectively, which can again be summed up to the total number of edges, $e = e_o + e_u$.
Recall is the fraction of edges in the \ac{cpdag} that are oriented:
\begin{equation}
    \text{Recall} = \frac{e_o}{e} \; .
\end{equation}

\section{Other causal methods used for comparison} \label{app:comparison_methods}

\subsection*{Granger-causality} We implement Granger-causality through software provided by \citet{seabold2010statsmodels} for $d=2$ and by \citet{runge2019detecting} for $d>2$, and define the optimal lag as the first local minimum in the mutual information function over measurement points. 
Given that Granger-causality is applied to a pair of samples, $(x_i, y_i)$, we iterate over every pair and determine the percentage of pairs that result in $X \rightarrow Y$ against $X \leftarrow Y$. 
We follow a similar process for more than two variables.

\subsection*{\ac{ccm}} We implement \ac{ccm} through software provided by \citet{Javier_causal-ccm_a_Python_2021} and define the optimal lag as for Granger-causality. 
The embedding dimensionality is determined by a false nearest neighbour method \citep{teaspoon_2022}. 
As with Granger-causality, \ac{ccm} is also applied to a pair of samples, $(x_i, y_i)$. We take the same summarising approach and iterate over every pair and determine the percentage of of pairs that result in $X \rightarrow Y$ against $X \leftarrow Y$.
\subsection*{\ac{pcmci}} \citet{runge2019detecting} proposed \ac{pcmci}, a method that produces \emph{time-series graphs} where each node represents a variable at a certain state in time.
In our experiments, we estimate the maximum possible lag as the average of the first local minimum of mutual information of all considered data series. 
To test for the presence of an edge, we apply distance correlation-based independence tests \citep{szekely2007measuring} between the residuals of $X$ and $Y$ that remain after regressing out the influence of the remaining nodes in the graph through Gaussian processes. 
The distance correlation coefficients of the significant edges between two variables are summed up to measure the prevalence of a causal direction of edges connecting these two variables.  
The causal direction with the greater sum is  considered to be the directed edge, and assessed against the edge between the same two variables in the true \ac{dag}.

\section{Computational complexity of our proposed methods}
\label{app:complexity}
\subsection*{Regression-based causal discovery} 
One needs to iterate over all possible candidate \acp{dag} to check the possibility that it was hte causal graph that generated the observed data. The total number of candidate \acp{dag} $\kappa_d$ of a graph with $d$ nodes is given by
\begin{equation}
    \kappa_d = \sum_{i=1}^d (-1)^{i+1} \binom{d}{i} 2^{i (d-i)} \kappa_{d-i} \; .
\end{equation}

\subsection*{Constraint-based approach} 
One needs to exhaustively test for conditional independence between any two nodes of the graph given any possible subset of the remaining nodes. Including the empty conditional set, this results in $2^{d-2}$ conditional independence tests for every pair of nodes in the graph. Then, the total number of conditional independence tests is:
\begin{equation}
    \eta_d = \binom{d}{2} 2^{d-2} \; .
\end{equation}

\section{Comparison of running times of our proposed method and the other causal methods} \label{app:running_time}

Table~\ref{tab:running_time} shows the running times of the proposed and comparison methods for synthetic data in Section~\ref{sec:experiments}.
Furthermore, for the bivariate real-world socioeconomic data set of \textit{corruption} and \textit{income inequality} in Section~\ref{sec:wgi_data},
Granger-causality completes the computation in 731 ms, \ac{ccm} in 537 ms, and the regression-based approach in 19.2 s (according to our Python implementations).

\begin{table}[htb!]
\footnotesize
    \centering
    \begin{tabular}{|l| l l l | l l l|}
         \hline
         $d$ & Regression-based & Constraint-based & Combined & Granger causality & CCM & PCMCI \\
         \hline
         2 & 1 min 57 s & -- & -- & 998 ms  & 1.57 s & -- \\
         3 & 30 min 35 s & 1 min 10 s & 3 min 13 s & 43 s & -- & 45.5 s \\
         4 & -- & 8 min 10 s & 8 min 1 s & 1 min 23 s & -- & 1 min 29 s \\
         5 & -- & 34 min 33 s & 39 min 8 s & 2 min 44 s & -- & 2 min 50 s \\
         6 & -- & 1 h 31 min 36 s & 1 h 43 min 54 s & 3 min 20 s & -- & 3 min 27 s \\ \hline
    \end{tabular}
    \caption{Running times for one trial with 100 samples for our three proposed methods (regression, constraint, combined) and the three comparison methods (Granger, \ac{ccm}, \ac{pcmci})   over various graph sizes (number of variables $d \in \{2, 3, 4\}$). The dash ``--'' implies that the method was not applied in our experiments in Section~\ref{sec:experiments}.}
    \label{tab:running_time}
\end{table}

\section{Influence of nonlinearity and nonstationarity on Granger causality and \ac{ccm}} \label{app:linear_stationary}
\subsection*{Nonlinearity and Granger causality:}
To examine the increased robustness of our proposed algorithm to nonlinearity, we define a bivariate system with $X$ as the cause and $Y$ as the effect under the assumption of causal sufficiency, i.e., $X \rightarrow Y$.
We consider $n$ functional samples $\{x_i(s)\}_{i=1}^n$ of $X$ defined over the interval $s = [0, 1]$ generated  as described in Equation \eqref{eq:func_data},
%
%
and the samples $\{y_i(t)\}_{i=1}^n$ of $Y$ defined over $t = [0, 1]$ by: 
\begin{equation}
    y_i (t) = (1 - r) \ \int_{0}^t x_i(s) \; ds + r \ \int_{0}^t x_i(s) \; \beta(s,t) \; ds \;.
\end{equation}
%
By sweeping $r \in [0,1]$, we go from linear dependence between $x_i(s)$ and $y_i(t)$ at $r = 0$ to nonlinear dependence at $r=1$, where we recover the model in Equation~\eqref{eq:hist_data} with $a=1$. 
We thus expect high accuracy from Granger-causality at $r=0$, where there is stationarity and linearity~\citep{granger1969investigating}.
As $r$ increases, we expect that our regression-based approach will achieve more accurate results, and the results of Granger will degrade as the assumption of linearity will not be met.
As expected, Figure~\ref{fig:linearity_granger} shows  the improved prediction accuracy (lower \ac{shd}) of our regression-based approach as $r$ is increased, reaching
\ac{shd} of $0.425 \pm 0.1542$ (versus $0.8 \pm 0.092$  
for Granger-causality) at $r=1$.

\begin{figure}[htb!]
    \centering
    \includegraphics[width=0.45\linewidth]{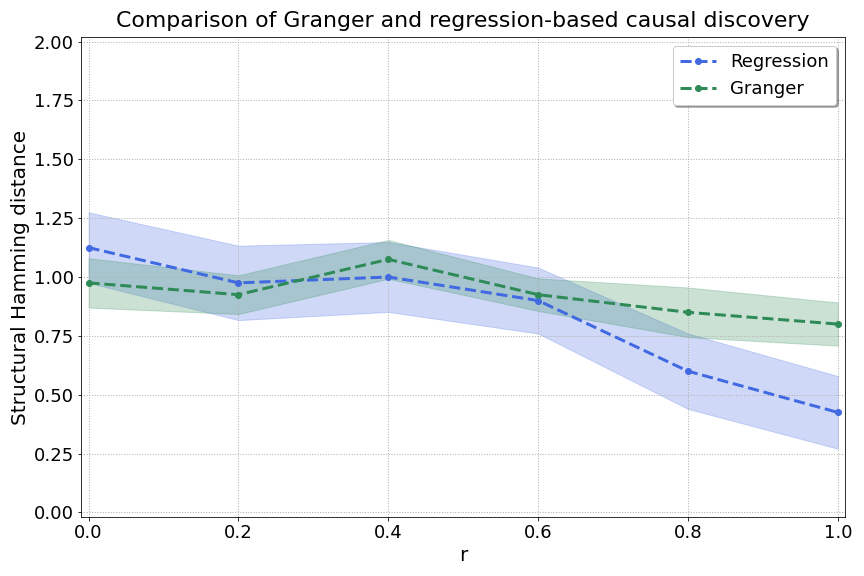}
    \caption{Comparison between Granger-causality and our regression-based method for $n = 100$ and various values for $r$ which influence the level of nonlinearity in the dependence from $X$ to $Y$.}
    \label{fig:linearity_granger}
\end{figure}

\subsection*{Nonstationarity and \ac{ccm}:}
To evaluate the importance of nonstationarity in causal learning, we compare our regression-based method against \ac{ccm}~\citep{sugihara2012detecting}, 
as a nonstationary component is introduced.
We generate samples 
from a bivariate coupled nonlinear logistic map as defined in \citet{sugihara2012detecting} such that $X$ influences $Y$ more strongly than vice-versa and we add nonstationary trends $\nu_X (t)$ and $\nu_Y (t)$:
\begin{align}
    x_i(t+1) &= x_i(t) \ \left[ 0.8 - 3.5 x_i(t) - 0.02 y_i(t) \right] + r \ \nu_X (t)  \\
    y_i(t+1) &= y_i(t) \ \left[ 0.2 - 3.2 y_i(t) - 0.1 x_i(t) \right] + r \ \nu_Y (t) \;,
\end{align}
where
$\nu_X (t) = \tanh \left (c_\nu t - \frac{c_\nu}{2} \right)$
with $c_\nu \sim \mathcal{N}(8, 1)$, 
and similarly for $\nu_Y (t)$.
The factor $r \in [0, 1]$ controls the weight of the nonstationarity of the time-series as it is increased.

As shown in Figure~\ref{fig:stationarity_ccm}, \ac{ccm} is only able to capture the true causal direction $X \rightarrow Y$ for $r = 0$, as expected, but \ac{ccm} loses its detection power
as soon as the data become nonstationary ($r>0$).  
In contrast, the regression-based maintains its detection power as the nonstationarity is increased.

\begin{figure}[htb!]
    \centering
    \includegraphics[width=0.45\linewidth]{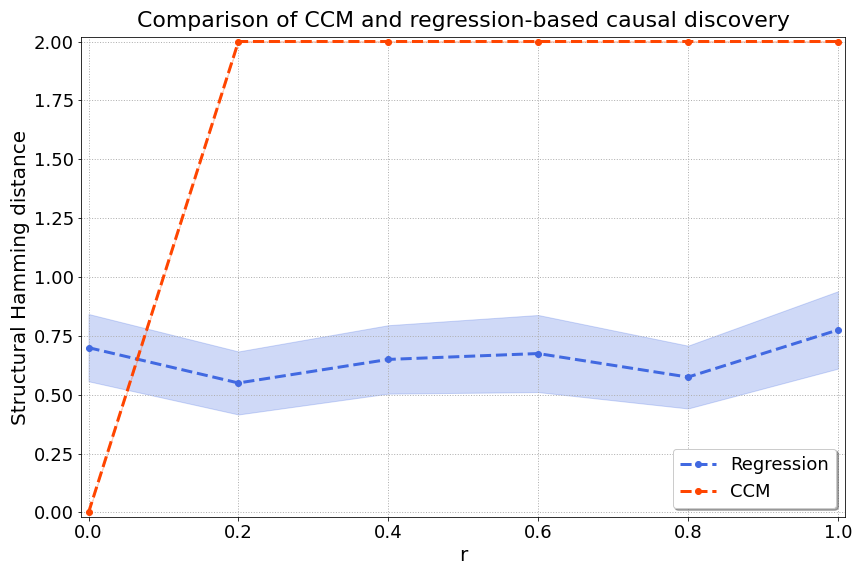}
    \caption{Comparison between \ac{ccm} and our regression-based method for $n = 100$ and various values for $r$ which influence the level of nonstationarity in the data for $X$ and $Y$.}
    \label{fig:stationarity_ccm}
\end{figure}


    
\newpage
\section{World Governance Indicators}\label{app:WGI}

In Figure~\ref{fig:WGI_data}, we use the official names, as mentioned in Table~\ref{tab:WGI_names}, where the Descriptions are taken from \url{https://info.worldbank.org/governance/wgi/}.
The data set consists of six variables (normalised to a range from 0 to 100), each of which is a summary of multiple indicators measured over 25 years from 1996 to 2020 in 182 countries, as described in~\citet[\S 4]{kaufmann2011worldwide}.

\begingroup
\renewcommand{\arraystretch}{1.5}
\begin{table}[h]
\newcolumntype{L}{>{\raggedright\arraybackslash}X}
\newcolumntype{l}{>{\raggedright\arraybackslash\hsize=.5\hsize}X}
    \centering
    \begin{tabularx}{\textwidth}{l l L}
         \toprule
         Abbreviation & Official Name & Description  \\ 
         \midrule
         VA & Voice and Accountability & Voice and accountability captures perceptions of the extent to which a country's citizens are able to participate in selecting their government, as well as freedom of expression, freedom of association, and a free media. \\
         PS & Political Stability and No Violence & Political Stability and Absence of Violence/Terrorism measures perceptions of the likelihood of political instability and/or politically-motivated violence, including terrorism. \\
         GE & Government Effectiveness & Government effectiveness captures perceptions of the quality of public services, the quality of the civil service and the degree of its independence from political pressures, the quality of policy formulation and implementation, and the credibility of the government's commitment to such policies.  \\
         RQ & Regulatory Quality & Regulatory quality captures perceptions of the ability of the government to formulate and implement sound policies and regulations that permit and promote private sector development. \\
         RL & Rule of Law & Rule of law captures perceptions of the extent to which agents have confidence in and abide by the rules of society, and in particular the quality of contract enforcement, property rights, the police, and the courts,  as well as the likelihood of crime and violence. \\
         CC & Control of Corruption & Control of corruption captures perceptions of the extent to which public power is exercised for private gain, including both petty and grand forms of corruption, as well as "capture" of the state by elites and private interests. \\ 
         \bottomrule
    \end{tabularx}
    \vspace{8pt}
    \caption{Abbreviations, official names and descriptions of the variables in the World Governance Indicators data set.}
    \label{tab:WGI_names}
\end{table}
\endgroup

\end{document}